\documentclass[11pt]{article}
\RequirePackage{algorithm,algorithmic,graphicx,amssymb,epsf,epic,url,complexity}
\RequirePackage{fullpage}
\usepackage{mdframed}
\usepackage{microtype, complexity}
\usepackage[usenames,dvipsnames]{xcolor}
\usepackage{amsthm}
\usepackage{amsmath}
\usepackage{upref}
\usepackage{amsfonts}
\usepackage{graphicx}
\usepackage{latexsym}
\usepackage{amssymb}
\usepackage{amscd}
\usepackage{mdframed}
\usepackage[all]{xy}
\usepackage{mathrsfs}
\usepackage{amsfonts}
\usepackage{epsfig}
\usepackage{epstopdf}
\usepackage{tcolorbox}
\newcommand {\ignore} [1] {}
\def\eps{\varepsilon}
\newcommand{\e}{{\varepsilon}}

\def\cM{{\cal M}}

\def\cS{{\cal S}}
\def\cC{{\cal C}}

\newcommand{\F}{{\mathcal{F}}}

\def\cA{{\cal A}}
\def\cG{{\cal G}}
\def\cB{{\cal B}}

\newcommand{\NULL}{\textsc{null}}

\def\denseformat{
\setlength{\textheight}{9.2in}
\setlength{\textwidth}{6.9in}
\setlength{\evensidemargin}{-0.2in}
\setlength{\oddsidemargin}{-0.2in}
\setlength{\headsep}{10pt}
\setlength{\topmargin}{-0.3in}
\setlength{\columnsep}{0.375in}
\setlength{\itemsep}{0pt}
}

\newtheorem{theorem}{Theorem}[section]

\newtheorem{lemma}[theorem]{Lemma}

\newtheorem{remark}[theorem]{Remark}

\def\boldhead#1:{\par\vskip 7pt\noindent{\bf #1:}\hskip 10pt}
\def\ithead#1:{\par\vskip 7pt\noindent{\it #1:}\hskip 10pt}

\def\inline#1:{\par\vskip 7pt\noindent{\bf #1:}\hskip 10pt}
\def\midinline#1:{\par\noindent{\bf #1:}\hskip 10pt}
\def\dnsinline#1:{\par\vskip -7pt\noindent{\bf #1:}\hskip 10pt}
\def\ddnsinline#1:{\newline{\bf #1:}\hskip 10pt}
\def\largeinline#1:{\par\vskip 7pt\noindent{\large\bf #1:}\hskip 10pt}
\long\def\commhide #1\commhideend{}
\long\def\commtim #1\commendt{#1}
\long\def\commb #1\commbend{}
\long\def\commedit #1\commeditend{} 

\long\def\commB #1\commBend{}       

\long\def\commex #1\commexend{}     

\long\def\commsiena #1\commsienaend{}  

\long\def\commBI #1\commBIend{}  

\long\def\CProof #1\CQED{}

\def\blackslug{\hbox{\hskip 1pt \vrule width 4pt height 8pt
    depth 1.5pt \hskip 1pt}}
\def\QED{\quad\blackslug\lower 8.5pt\null\par}

\long\def\PPP#1{\noindent{\bf Proof:}{ #1}{\quad\blackslug\lower 8.5pt\null}}

\long\def\denspar #1\densend
{#1}

\newif\ifnotesw\noteswtrue
   {\ifnotesw\marginpar[\hfill\(\top\)]{\(\top\)}\fi}%
   {\ifnotesw\marginpar[\hfill\(\bot\)]{\(\bot\)}\fi}

\newcommand{\mnote}[1]%
    {\ifnotesw\marginpar%
        [{\scriptsize\it\begin{minipage}[t]{\marginparwidth}
        \raggedleft#1%
                        \end{minipage}}]%
        {\scriptsize\it\begin{minipage}[t]{\marginparwidth}
        \raggedright#1%
                        \end{minipage}}%
    \fi}

\def\cA{{\cal A}}
\def\cB{{\cal B}}
\def\cC{{\cal C}}

\def\cF{{\cal F}}
\def\cG{{\cal G}}

\def\cM{{\cal M}}

\def\cS{{\cal S}}

\def\MathF{\hbox{\rm I\kern-2pt F}}
\def\MathP{\hbox{\rm I\kern-2pt P}}
\def\MathR{\hbox{\rm I\kern-2pt R}}
\def\MathZ{\hbox{\sf Z\kern-4pt Z}}
\def\MathN{\hbox{\rm I\kern-2pt I\kern-3.1pt N}}
\def\MathC{\hbox{\rm \kern0.7pt\raise0.8pt\hbox{\footnotesize I}
\kern-4.2pt C}}
\def\MathQ{\hbox{\rm I\kern-6pt Q}}

\newsavebox{\ttop}\newsavebox{\bbot}

\def\eps{\epsilon}

\def\polylog{\mbox{polylog}}

\denseformat
\newtheorem{mdresult}[theorem]{Theorem}
\newenvironment{Theorem}{\begin{mdframed}[backgroundcolor=lightgray!40,topline=false,rightline=false,leftline=false,bottomline=false,innertopmargin=2pt]\begin{mdresult}}{\end{mdresult}\end{mdframed}}

\newtheorem{question}{Question}
\newtheorem{mdresult2}[question]{Question}
\newenvironment{Question}{\begin{mdframed}[backgroundcolor=lightgray!40,topline=false,rightline=false,leftline=false,bottomline=false,innertopmargin=2pt]\begin{mdresult2}}{\end{mdresult2}\end{mdframed}}

\title{A Generalized Matching Reconfiguration Problem} 
\author{
Noam Solomon \thanks{Center for Mathematical Sciences and Applications, Harvard University, Cambridge, Massachusetts, USA. Email: \texttt{noam.solom@gmail.com}.}
\and Shay Solomon\thanks{School of Electrical Engineering, Tel Aviv University. Email: {\small \texttt{shayso@post.tau.ac.il.}}}}

\date{}
\begin{document}

\maketitle
\thispagestyle{empty}
\begin{abstract}
The goal in {\em reconfiguration problems} is to compute a {\em gradual transformation} between two feasible solutions of a problem such that all intermediate solutions are also feasible.
In the {\em Matching Reconfiguration Problem} (MRP), proposed in a pioneering work by Ito et al.\ from 2008, 
we are given a graph $G$   
and two matchings $M$ and $M'$,    
and we are asked whether there is a sequence of matchings in $G$ starting with $M$ and ending at $M'$, 
each resulting from the previous one by either adding or deleting a single edge in $G$, 
without ever going through a matching of size $< \min\{|M|,|M'|\}-1$.
Ito et al.\ gave a polynomial time algorithm for the problem, which uses the Edmonds-Gallai decomposition.

In this paper we introduce a natural generalization of the MRP that depends on   
an integer parameter $\Delta \ge 1$:   
here we are allowed to make $\Delta$ changes to the current solution rather than 1 at each step of the {transformation procedure}.
There is always a valid sequence of matchings transforming $M$ to $M'$ if $\Delta$ is sufficiently large, and naturally we would like to minimize $\Delta$.
We first devise an optimal transformation procedure for unweighted matching with $\Delta = 3$, 
and then extend it to weighted matchings to achieve asymptotically optimal guarantees.
The running time of these procedures is linear.   

We further demonstrate the applicability of this generalized problem to dynamic graph matchings.
In this area, the number of changes to the maintained matching per update step (the \emph{recourse bound}) is an important quality measure.
Nevertheless, the \emph{worst-case} recourse bounds of almost all known dynamic matching algorithms are prohibitively large, much larger than the corresponding update times.
We fill in this gap via a surprisingly simple black-box reduction:
Any dynamic algorithm for maintaining a $\beta$-approximate maximum cardinality matching
with update time $T$, for any $\beta \ge 1$, $T$ and $\eps > 0$,
can be \emph{transformed} into an algorithm for maintaining a $(\beta(1 +\eps))$-approximate maximum cardinality matching with update time $T + O(1/\eps)$ and worst-case recourse bound $O(1/\eps)$.
This result generalizes for approximate maximum weight matching, where the update time and worst-case recourse bound grow from
$T + O(1/\eps)$ and $O(1/\eps)$ to $T + O(\psi/\eps)$ and $O(\psi/\eps)$, respectively; $\psi$ is the graph {\em aspect-ratio}.
We complement this positive result by showing that, for $\beta = 1+\eps$, the worst-case recourse bound of any algorithm produced by our reduction is optimal.
As a corollary, several key dynamic approximate matching algorithms--- with 
poor worst-case recourse bounds---
are strengthened to achieve near-optimal worst-case recourse bounds with no loss in update time.
\end{abstract}

\clearpage
\setcounter{page}{1}
\section{Introduction}

The study of graph algorithms is mostly concerned with the measure of \emph{(static) runtime}.
Given a graph optimization problem, the standard objective is to design a fast (possibly approximation) algorithm, and ideally complement it with a matching lower bound on the runtime of any (approximation) algorithm
for solving the problem.
As an example, computing (from scratch) a 2-approximate minimum vertex cover (VC) can be done trivially in linear time,
whereas a better-than-2 approximation for the minimum VC cannot be computed in polynomial time under the unique games conjecture \cite{KR08}.

The current paper is motivated by a natural need arising in networks that are prone to temporary or permanent changes.
Such changes are sometimes part of the normal behavior of the network, as in {\em dynamic  networks},
but changes could also be the result of unpredictable failures of nodes and edges, particularly in \emph{faulty networks}.
Consider a large-scale network $G = (V,E,w)$ for which we need to solve, perhaps approximately, some graph optimization problem,
and the underlying solution (e.g., a maximum matching) is being used for some practical purpose (e.g., scheduling in packet switches) throughout a long time span. 
If the network changes over time, the quality of the used solution may degrade   until it is too poor to be used in practice
and it may even become infeasible. 

Instead of the standard objectives of optimization, the questions that arise here concern \emph{reoptimization}:
Can we ``efficiently'' transform one given solution (the \emph{source}) to another one  (the \emph{target}) under ``real-life constraints''?
The efficiency of the {\em transformation procedure} could be measured in terms of running time, but in some applications making even small changes to the currently used solution may incur huge costs,
possibly much higher than the runtime cost of computing (from scratch) a better solution;
we shall use ``procedure'' and ``process'' interchangeably.
In particular, this is often the case whenever the edges of the currently used solution are ``hard-wired'' in some physical sense, as in road networks.
Various real-life constraints or objectives may be studied; the one we focus on in this work 
is that at any step (or every few steps) throughout the transformation process the current solution should be both feasible and of quality no worse (by much) than that of either the source or   target solutions.
This constraint is natural as it might be prohibitively expensive or even impossible
to carry out the transformation process \emph{instantaneously}.
Instead, the transformation can be broken into \emph{phases} each performing $\le \Delta$ changes to the transformed solution, where $\Delta \ge 1$ is some parameter,
so that the solution obtained at the end of each phase--- to be used instead of the source solution--- is both feasible and of quality no (much) worse than either the source or target.
The transformed solution is to eventually coincide with the target solution.

The arising \emph{reoptimization} meta-problem generalizes the well-studied framework of {\emph reconfiguration problems}, which we discuss in Section 1.1. 
It is interesting from both practical and theoretical perspectives,
since 
even the most basic and well-understood optimization problems become open in this setting.
E.g., for the VC problem,
{\em given} a better-than-2 approximate target VC,
can we transform to it from any source VC subject to the above constraints?
This is an example for a problem that is computationally hard in the standard sense but might be easy from a reoptimization perspective.
In contrast, perhaps computationally easy problems, such as approximate maximum matching, are hard from a reoptimization perspective?

This meta-problem captures tension between (1) the {\em global} objective of transforming one global solution to another,
and (2) the  {\em local} objective of transforming \emph{gradually} while having a feasible and high quality solution throughout the process.
A similar tension is captured by various models of computation that involve locality,
including dynamic graph algorithms, distributed computing, property testing and local computation algorithms (LCA).
The study of the meta-problem presented here could borrow from these related research fields, but,   more importantly, 
we anticipate that it will also contribute to them; indeed, we present here an application of this meta-problem to dynamic graph algorithms.  
\vspace{7pt}
\\
\noindent
{\bf 1.1~ Graph Reconfiguration.~}
The framework of {\em reconfiguration problems} has been subject to growing interest in recent years.
The term {\em reconfiguration} was coined in the work of Ito et al.\ 
\cite{DBLP:journals/tcs/ItoDHPSUU11}, which unified earlier related problems and terminology 
(see, e.g.,  \cite{DBLP:conf/icalp/HearnD02,gopalan2009connectivity,bonsma2009finding})
into a single framework. 
The general goal is to compute a {\em transformation} between two feasible solutions of a problem such that all intermediate solutions are also feasible, where each pair of consecutive solutions need to be {\em adjacent} under a fixed polynomially testable symmetric adjacency relation on the set of feasible solutions.
Such a transformation arises naturally in many contexts, such as solving puzzles, motion planning, 
questions of evolvability (can genotype evolve into another one via individual ``adjacent'' mutations?), 
and similarity of DNA sequences in computational genomics and particularly gene editing, which is among the hottest scientific topics these days; 
see the surveys of \cite{DBLP:books/cu/p/Heuvel13,DBLP:journals/algorithms/Nishimura18} for further details.
In most previous work, two solutions are called adjacent if their symmetric difference has size 1.
The most well-studied problem under this framework is graph matching.  
For brevity, we shall only discuss here papers on graph matching; see the surveys  \cite{DBLP:books/cu/p/Heuvel13,DBLP:journals/algorithms/Nishimura18} for discussions on other problems.

In the {\em Matching Reconfiguration Problem} (MRP), proposed in \cite{DBLP:journals/tcs/ItoDHPSUU11}, 
we are given a graph $G$ 
and two matchings $M$ and $M'$,   
and we are asked whether there is a sequence of matchings in $G$ starting with $M$ and ending at $M'$, 
each resulting from the previous one by either adding or deleting a {\em single edge} in $G$, 
without ever going through a matching of size $< \min\{|M|,|M'|\}-1$.
Ito et al.\ gave a polynomial time algorithm for the problem, 
which uses the Edmonds-Gallai decomposition. In particular, in some cases such a transformation does not exist,
and much of the difficulty  is in the decision problem (decide if exists or not).  
The problem of generalizing this algorithm for weighted matchings was proposed as an open problem in \cite{DBLP:journals/tcs/ItoDHPSUU11}, and remained open to date, partially
since the algorithm of \cite{DBLP:journals/tcs/ItoDHPSUU11} already for unweighted matchings relies on a rather intricate decomposition.  
The work of  \cite{DBLP:journals/tcs/ItoDHPSUU11} triggered interesting followups on MRP
\cite{DBLP:journals/tcs/KaminskiMM12,DBLP:journals/jco/ItoKKKO19,gupta2019complexity,DBLP:conf/esa/ItoKK0O19,DBLP:conf/mfcs/BonamyBHIKMMW19,DBLP:conf/wg/BousquetHIM19}. In all these followups, 
the symmetric difference between two adjacent matchings is rather strict: it is fixed by either 1 or 2 in 
\cite{DBLP:journals/tcs/KaminskiMM12,DBLP:journals/jco/ItoKKKO19,gupta2019complexity,DBLP:conf/wg/BousquetHIM19}, 
whereas in the context of perfect matchings the symmmetric difference is an alternating cycle of length 4
\cite{DBLP:conf/mfcs/BonamyBHIKMMW19,DBLP:conf/esa/ItoKK0O19}. 
Perhaps since the symmetric difference in all the previous work is so strict, the goal was polynomial-time algorithms and hardness for the problem.
The natural generalization of parameterizing the symmetric difference by an arbitrary $\Delta, \Delta \ge 1$--- as in our reoptizmiation meta-problem--- was not studied in prior work.  
\vspace{7pt}
\\
\noindent
{\bf 1.2~ Our contribution.~} 
We study two fundamental graph matching problems under the aforementioned meta-problem:
 (approximate) maximum cardinality matching (MCM) and maximum weight matching (MWM).
Our meta-problem is, in fact, inherently different than the original MRP.
We are not interested in the decision version of the problem--- we take $\Delta$ to be large enough so that a  transformation 
is   {\em guaranteed} to exist. Thus we shift the focus from per-instance optimization to {\em existential optimization}, and our goal is to optimize $\Delta$ so that any source matching can be transformed to any target matching
by performing at most $\Delta$ changes per step, while never reaching a {\em much worse} matching than 
either the source or the target along the way.
By ``worse'' we mean either in terms of size or weight, and we must indeed do a bit worse in some cases even for large $\Delta$; the original MRP formulation for unweighted matchings allows to go down by 1 unit of size, and this slack is required also for large $\Delta$.
For weighted graphs, naturally, a bigger slack is required.
For both unweighted and weighted matchings, we provide transformation procedures
with near-optimal guarantees and linear running time.  
Our results are summarized next; the transformation for approximate MWM (Theorem~\ref{th:main}) is the most technically challenging part of this work.


\begin {Theorem} [MCM]
\label{th:MCM}
For any source and target matchings $\cM$ and $\cM'$,
one can transform $\cM$ into (a possibly superset of) $\cM'$ via a sequence of phases consisting of $\le$ 3 operations each (i.e., $\Delta = 3$), such that the matching   at the end of each phase
throughout this transformation is a valid matching for $G$ of size $\ge \min\{|\cM|,|\cM'|-1\}$.
The runtime of this transformation procedure is  $O(|\cM| + |\cM'|)$.
\end {Theorem}


\begin {Theorem} [MWM]
\label{th:main}
For any source and target matchings $\cM$ and $\cM'$ with $w(\cM') > w(\cM)$, and   any $\eps >0$, one can   transform $\cM$ into (a possibly superset of) $\cM'$ via a sequence of phases consisting of $O(\frac 1 \eps)$ operations each
(i.e., $\Delta = O(\frac 1 \eps)$), such that the matching obtained at the end of each phase
throughout this transformation   is a valid matching for $G$ of weight $\ge \max\{w(\cM)-W,(1-\eps)w(\cM)\}$, where $W = \max_{e \in \cM} w(e)$. 
The runtime of this transformation procedure is $O(|\cM| + |\cM'|)$.
\end {Theorem}
{\bf Remark.}
Theorem \ref{th:main} assumes  that $w(\cM') > w(\cM)$.
This assumption is made without loss of generality, since, if $w(\cM') \le w(\cM)$, we can apply a reversed transformation, so that the matching will always be of weight $\ge
\max\{w(\cM')-W',(1-\eps)w(\cM')\}$, where $W' = \max_{e \in \cM'} w(e)$.
\ignore{
\vspace{-8pt}
\begin {Theorem} [MSF]
\label{forest}
For any source and target spanning forests $\cF$ and $\cF'$,
one can gradually transform $\cF$ into $\cF'$ via a sequence of constant-time operations,
grouped into phases   of two operations each, such that the spanning forest obtained at the end of each phase
throughout this transformation process is a valid spanning forest for $G$ of weight at most $\max\{w(\cF),w(\cF')\}$.
Moreover, the runtime is  $O\bigl((|\F|+|\F'|)\log(|\F|+|\F'|)\bigr)$.
\end {Theorem}
}

In App.\ \ref{tightness}, we show that the guarantees provided by Theorems \ref{th:MCM} 
and \ref{th:main} are tight and asymptotically tight, respectively.
Although our   results may lead to the impression that there exists an efficient gradual transformation process to any  graph optimization problem,
we briefly discuss in App.\ \ref{discuss} two trivial hardness results for the minimum VC and maximum independent set problems. 
\vspace{7pt}
\\
\noindent
{\bf 1.2.1~ Application: A worst-case recourse bound for dynamic matching algorithms.}
In the standard \emph{fully dynamic} setting we start from an empty graph  $G_0$ on $n$ fixed vertices,
and at each time step $i$ a single edge $(u,v)$ is either inserted to the graph $G_{i-1}$ or deleted from it, resulting in graph $G_i$.
In the {\em vertex update} setting we have vertex updates instead of   edge updates;
this setting was mostly studied for bipartite graphs  \cite{BLSZ14,BLSZ15,BHR17}.

The problem of maintaining a large matching in fully dynamic graphs was subject to intensive interest recently \cite{OR10,BGS11,NS13,GP13, PS16,Sol16,BHN17,BHR17,CS18,ACCSW18,GLSSS19,BFH19}.
The basic goal is to devise an algorithm for maintaining a large matching while keeping a  tab on the \emph{update time},
i.e., the time required to update the matching at each step.
One may try to optimize the \emph{amortized} (average) update time of the algorithm or its \emph{worst-case} (maximum) update time,
but both measures are defined with respect to a \emph{worst-case} sequence of graphs.

 ``Maintaining'' a matching with update time $u_T$ translates into maintaining a data structure with update time $u_T$,
which answers queries regarding the matching with a low, ideally constant, \emph{query time} $q_T$.
For a queried vertex $v$ the answer is the only matched edge incident on $v$, or $\NULL$ if $v$ is free,
while for a queried edge $e$ the answer is whether edge $e$ is matched or not.
All queries made following the same update step $i$ should be answered \emph{consistently} with respect to the same matching, hereafter the \emph{output matching (at step $i$)},
but queries made in the next update step $i+1$ may be answered with respect to a completely different matching.
Thus
even if the worst-case update time is   low,
the output matching
may change significantly from one update step to the next;  
some natural scenarios where the output matching changes significantly per update step are discussed in App.\ \ref{12}.



The number of changes (or replacements) to the output matching per update step is an important measure of quality,
sometimes referred to as the \emph{recourse bound},
and the problem of optimizing it has received growing attention recently \cite{GKKV95,CDKL09,GKS14,BLSZ15,BLZS17,BHR17,BKPPS17,ADJ18,MSV18}.
 In applications such as job scheduling, web hosting, streaming content delivery, data storage and hashing, a replacement of a matched edge by another one may be   costly, possibly much more than the runtime of computing these replacements.
 Moreover, when the recourse bound is low, one can efficiently \emph{output} all the changes to the matching following every update step, which could be important in practical scenarios. In particular, a low recourse bound is important when the matching algorithm is used as a black-box subroutine inside a larger data structure or algorithm \cite{BS16,ADKKP16}; see App.\ \ref{123} for more details. We remark that the recourse bound (generally defined as the number of changes to some underlying structure per update step) has been well studied  in the areas of dynamic and online algorithms for a plethora of optimization problems besides graph matching, such as MIS, set cover, Steiner tree, flow and scheduling; see \cite{GGK13,GK14,GKS14,BGKPSS15,MSVW16,AOSS18,CHK16,GKKP17,SSTT18},
 and the references therein.

There is a strong separation between the state-of-the-art amortized versus worst-case bounds
for dynamic matching algorithms, in terms of both the time and the recourse bounds.
A similar separation exists for numerous other problems, such as dynamic 
minimum spanning forest. 
In various practical scenarios, particularly in systems designed to provide real-time responses,
a strict tab on the \emph{worst-case update time} or on the \emph{worst-case recourse bound} is crucial, thus an algorithm with a low amortized guarantee but a high worst-case guarantee is  useless.

Despite the importance of the recourse bound measure,  all known algorithms but one in the area of dynamic matchings (described in App.\ \ref{111})
provide no nontrivial worst-case recourse bounds whatsoever! The sole exception is an algorithm 
for maintaining a maximal matching with a worst-case update time $O(\sqrt{m})$ and a constant recourse bound \cite{NS13}.
In this paper we fill in this gap via a surprisingly simple 
yet powerful black-box reduction (throughout \emph{$\beta$-MCM} is a shortcut for $\beta$-approximate MCM):
\begin{Theorem} \label{main}
Any dynamic algorithm  maintaining a $\beta$-MCM
with update time $T$,\footnote{Besides
answering queries, we naturally assume that at any update step 
the entire matching  can be output within time (nearly) linear in its size.
All known algorithms satisfy this assumption.}
 for any $\beta \ge 1$, $T$ and $\eps > 0$,
can be \emph{transformed} into an algorithm   maintaining a $(\beta(1 +\eps))$-MCM with  update time  $T + O(1/\eps)$ and
worst-case recourse bound  $O(1/\eps)$. If the original time bound $T$ is amortized/worst-case, so is the resulting time bound of $T + O(1/\eps)$,
while the recourse bound $O(1/\eps)$ always holds in the worst-case.
This applies to the fully dynamic setting  under edge and/or vertex updates.  
\end{Theorem}

The proof of Theorem \ref{main} is carried out in two steps.
First we prove Theorem \ref{th:MCM} by showing a simple transformation process for any two matchings $\cM$ and $\cM'$ of the same \emph{static} graph.
The second step of the proof, which is   the key insight behind it, is that the gradual transformation process can be used \emph{essentially as is} in fully dynamic graphs, 
while incurring a negligible loss to the size and approximation guarantee of the transformed matching.  

In Section \ref{recourse} we complement the positive result provided by Theorem \ref{main} by proving that the recourse bound $O(1/\eps)$ is optimal (up to a constant factor) in the regime $\beta = 1+\eps$.
In fact, the  lower bound $\Omega(1/\eps)$ on the recourse bound holds even in the amortized sense and even in the incremental (insertion only) and decremental (deletion only) settings.
For larger values of $\beta$,   taking $\eps$ to be a sufficiently small constant gives rise to an approximation guarantee arbitrarily close to $\beta$ with a
constant recourse bound.
\vspace{7pt} \\
{\bf A corollary of Theorem \ref{main}.~}
As a  corollary of Theorem \ref{main}, all previous algorithms \cite{GP13,BS15,PS16,CS18,ACCSW18,bernstein2019deamortization,DBLP:journals/corr/abs-1911-05545} with low worst-case update time are strengthened
to achieve a worst-case recourse bound of $O(1/\eps)$ with only an additive overhead of $O(1/\eps)$ to the update time.
(Some of these results were already strengthened in this way by using a previous version of the current work, which was posted to arXiv in 2018.)
Since the update time of all these algorithms is larger than $O(1/\eps)$,  we get a recourse bound of $O(1/\eps)$ with no loss whatsoever in the update time!
Moreover, all known algorithms with low amortized update time can be strengthened in the same way;
e.g., 
in SODA'19 \cite{GLSSS19} (cf.\ \cite{BLSZ14}) it was shown that one can maintain a $(1+\eps)$-MCM in the incremental edge update setting 
with a constant (depending exponentially on $\eps$) amortized update time.
While this algorithm yields a constant amortized recourse  bound, no nontrivial (i.e., $o(n)$) worst-case recourse bound was known for this problem.
Theorem \ref{main} strengthens the result of \cite{GLSSS19} to maintain a $(1+\eps)$-MCM with a constant amortized update time and the optimal worst-case recourse bound of $O(1 / \eps)$. 
Since the recourse bound is an important measure of quality, this provides a significant contribution to the area of dynamic matching algorithms.
\noindent \vspace{6pt} \\
{\bf Weighted matchings.~}
The result of Theorem \ref{main} can be generalized for approximate MWM in graphs with bounded aspect ratio $\psi$, 
by using the much more intricate transformation provided by Theorem \ref{th:main} (compared to Theorem \ref{th:MCM}), as summarized in the next theorem.
(The \emph{aspect ratio} $\psi = \psi(G)$ of a weighted graph $G=(V,E,w)$ is defined as $\psi = \frac{\max_{e \in E} w(e)} {\min_{e \in E} w(e)}$.)
\begin{Theorem} \label{main2}
Any dynamic algorithm for maintaining a $\beta$-approximate MWM (shortly, \emph{$\beta$-MWM})
with update time $T$ in a dynamic graph with aspect ratio always bounded by $\psi$, for any $\beta \ge 1$, $T, \eps > 0$ and $\psi$,
can be \emph{transformed} into an algorithm for maintaining a $(\beta(1 +\eps))$-MWM with update time $T + O(\psi/\eps)$ and
worst-case recourse bound $O(\psi/\eps)$. If the original time bound $T$ is amortized/worst-case, so is the resulting time bound of $T + O(\psi/\eps)$,
while the recourse bound $O(\psi/\eps)$ always holds in the worst-case.
This applies to the fully dynamic setting under   edge and/or vertex updates.
\end{Theorem}
\vspace{3pt}  \noindent 
{\bf Scenarios with high recourse bounds.~}
There are various scenarios where 
high recourse bounds may naturally arise. In   such scenarios our reductions (Theorems \ref{main} and \ref{main2}) can come into play to achieve low worst-case recourse bounds.
Furthermore, although a direct application of our reductions may only hurt the update time, we demonstrate the usefulness of these reductions in achieving low update time bounds in some natural settings (where we might not care at all about recourse bounds); this, we believe, provides another strong motivation for our reductions.
The details are deferred to App.\ \ref{12}.
\vspace{8pt}\\
\noindent
{\bf 1.3~ Related work.~}
We discussed in Section 1.1 prior work on graph reconfiguration problems.
Other than this line of work, there are also inherently different lines of work on ``reoptimiziation'', which indeed can be interpreted broadly--- there is an extensive and diverse body of research devoted to various notions of reoptimization; see \cite{thiongane2006,AusielloEMP09,BoriaP10,BiloBKKMSZ11,ausiello2011complexity,bender2015reallocation,bender2017cost,SSTT18,Bilo18}, and the references therein. The common goal in all previous work on reoptimization (besides the one discussed in Section 1.1 on reconfiguration) is to (efficiently) compute an exact or approximate solution to a new problem instance by using the solution for the old instance,
where typically the solution for the new instance should be close to the original one under certain distance measure. 
Our work is inherently different than all such previous work, since our starting point is that \emph{some solution to the new problem instance is given}, and the goal is to compute a \emph{gradual transformation process} (subject to some   constraints) between the two given solutions. Also, our work is inherently different than previous work on reconfiguration, as explained in Section 1.2.
\vspace{7pt}\\
\noindent
{\bf 1.4~ Organization.~}
This extended abstract (Sections \ref{changes}--\ref{recourse}) 
focuses on providing a tight worst-case recourse bound for dynamic approximate maximum matching.
We start (Section \ref{changes}) by describing a basic
scheme for dynamic approximate matchings that was introduced in \cite{GP13}.
In Section 3.1 we present a simple transformation process for MCM in static graphs, thus proving Theorem \ref{th:MCM}.
This result is generalized for MWM via a much more intricate transformation process that proves Theorem \ref{th:main},
which is deferred to App.\ \ref{sec:wema} due to space constraints.  
These transformations, which  apply to static graphs, are adapted to the fully dynamic setting in Sections 3.2 and 3.3, thus proving Theorems \ref{main} and \ref{main2}, respectively.
Our lower bound of $\Omega(1/\eps)$ on the recourse bound of $(1+\eps)$-MCMs is provided in Section \ref{recourse}.
A discussion is deferred to App.\ \ref{discuss}.  

\section{The scheme of \cite{GP13}} \label{changes}
This section provides a short overview of a basic scheme for dynamic approximate matchings from \cite{GP13}.
Although such an overview is not required for proving Theorems \ref{main} and \ref{main2}, it is instructive to provide it,
as it shows that the scheme of \cite{GP13} is insufficient for providing any nontrivial worst-case recourse bound.
Also, the scheme of \cite{GP13} exploits a basic \emph{stability} property of matchings, which we use for proving Theorems \ref{main} and \ref{main2},
thus an overview of this scheme may facilitate the  understanding of our proof.
\noindent \vspace{10pt} \\
{\bf 2.1~ The amortization scheme of \cite{GP13}.~} 
The \emph{stability} property of matchings used in \cite{GP13} is that the maximum matching size changes by at most 1 following each update step.
Thus if we have a $\beta$-MCM, for any $\beta \ge 1$,
 the approximation guarantee of the matching will remain close to $\beta$ throughout a long update sequence.
 Formally, the following lemma is a simple adaptation of Lemma 3.1 from \cite{GP13}; its proof is given in Appendix \ref{app:lazy} for completeness.
  (Lemma 3.1 of \cite{GP13} is stated for approximation guarantee $1+\eps$ and for edge updates, whereas Lemma \ref{lazylemma2} here holds for any approximation guarantee and also for vertex updates.)
  \ignore{
  \begin{lemma}  \label{lazylemma}
Let $\eps,\eps' \le 1/2$. Suppose that $\cM_i$ is a $(1+\eps)$-MCM for $G_i$.
For $j = i,i+1,\ldots, i+\lfloor \eps'\cdot |\cM_i| \rfloor$, let $\cM^{(j)}_i$ denote the matching $\cM_i$ after removing from it all edges that got deleted during the updates $i+1,\ldots,j$.
Then $\cM^{(j)}_i$ is a $(1+2\eps+2\eps')$-MCM for the graph $G_j$.
\end{lemma}
}

\ignore{
\subsection{An Adjustment to Lemma \ref{lazylemma}}
Lemma \ref{lazylemma} applies to approximation guarantees close to 1 and to the standard edge update setting.
However, as we show next, it is straightforward to extend it to any approximation guarantee $\beta$, in both the edge update and the vertex update settings.
}

\begin{lemma} \label{lazylemma2}
Let $\eps' \le 1/2$. Suppose  $\cM_t$ is a $\beta$-MCM for   $G_t$, for any $\beta \ge 1$.
For $i = t,t+1,\ldots, t+\lfloor \eps'\cdot |\cM_t| \rfloor$, let $\cM^{(i)}_t$ denote the matching $\cM_t$ after removing from it all edges that got deleted during updates $t+1,\ldots,i$.
Then $\cM^{(i)}_t$ is a $(\beta(1 +2\eps'))$-MCM for $G_i$.
\end{lemma}

For concreteness, we shall focus on the regime of approximation guarantee $1+\eps$,
and sketch the argument of \cite{GP13} for maintaining a $(1+\eps)$-MCM in fully dynamic graphs.
(As Lemma \ref{lazylemma2} applies to any approximation guarantee $\beta \ge 1+\eps$, it is readily verified that the same argument carries over to any approximation guarantee.)

One can compute a $(1+\eps/4)$-MCM $\cM_t$ at a certain update step $t$, and then re-use the same matching $\cM^{(i)}_t$ throughout all update steps $i = t,t+1,\ldots, t' =  t+ \lfloor \eps/4\cdot |\cM_t| \rfloor$ (after removing from it all edges that got deleted from the graph between steps $t$ and $i$).
By Lemma \ref{lazylemma2}, assuming $\eps \le 1/2$, $\cM^{(i)}_t$ provides a $(1+\eps)$-MCM for all graphs $G_i$.
Next compute a fresh $(1+\eps/4)$-MCM $\cM_{t'}$ following update step $t'$ and re-use it throughout all update steps $t',t'+1,\ldots,t'+ \lfloor \eps/4 \cdot |\cM_{t'}|\rfloor$, and repeat.
In this way the static time complexity of computing a $(1+\eps)$-MCM $\cM$
is \emph{amortized} over $1 + \lfloor \eps/4\cdot |\cM| \rfloor  = \Omega(\eps \cdot |\cM|)$ update steps.
As explained in Appendix \ref{amort}, the static computation time of an approximate matching is $O(|\cM| \cdot \alpha/\eps^2)$,
where $\alpha$ is the arboricity bound.  
(This bound on the static computation time was established in \cite{PS16}; it reduces to $O(|\cM| \cdot \sqrt{m}/\eps^2)$
and $O(|\cM| \cdot {\Delta}/\eps^2)$ for general graphs and graphs of degree bounded by $\Delta$, respectively, which are the bounds provided by \cite{GP13}.)
\noindent \vspace{10pt} \\
{\bf 2.2~ A Worst-Case Update time.~} 
In the amortization scheme of \cite{GP13} described above, a $(1+\eps/4)$-MCM $\cM$ is computed \emph{from scratch}, and then being re-used throughout $\lfloor \eps/4 \cdot |\cM| \rfloor$ additional update steps. The worst-case update time is thus the static computation time of an approximate matching, namely, $O(|\cM| \cdot \alpha/\eps^2)$.
To improve the worst-case guarantee, the tweak used in \cite{GP13} is to simulate the static approximate matching computation within
a ``time window'' of $1+\lfloor \eps/4 \cdot |\cM| \rfloor$  consecutive update steps, so that following each update step the algorithm simulates only
$O(|\cM| \cdot \alpha/\eps^2) / (1+\lfloor \eps/4 \cdot |\cM| \rfloor = O(\alpha \cdot \eps^{-3})$
steps of the static computation.
During this time  window the gradually-computed matching, denoted by $\cM'$, is useless, so the previously-computed matching $\cM$ is re-used as the output matching.
This means that each matching is re-used throughout a time window of twice
as many update steps, hence the approximation guarantee increases from $1+\eps$ to $1+2\eps$, but we can reduce it back to $1+\eps$ by a straightforward scaling argument.
Note that the gradually-computed matching does not include edges that got deleted from the graph during the time window.
\noindent \vspace{10pt} \\
{\bf 2.3~ Recourse bounds.~}
Consider an arbitrary time window used in the amortization scheme of \cite{GP13},
and note that the same matching is being re-used throughout the entire window.
 Hence there are no changes to the matching in the ``interior'' of the window except for those triggered by adversarial deletions,
which may trigger at most one change to the matching per update step.
On the other hand, at the start of any time window (except for the first), the output matching is switched from the old matching $\cM$ to the new one $\cM'$,
which may require $|\cM| + |\cM'|$ replacements to the output matching at that time.
Note that the amortized number of replacements per update step is quite low, being upper bounded by $(|\cM| + |\cM'|) / (1+\lfloor \eps/4 \cdot |\cM| \rfloor)$.
In the regime of approximation guarantee $\beta = O(1)$, we have $|\cM| = O(|\cM'|)$, hence the amortized recourse bound is bounded by $O(1/\eps)$.
For a general approximation guarantee $\beta$, the naive amortized recourse bound is $O(\beta/\eps)$.

On the negative side, the worst-case recourse bound may still be as high as $|\cM| + |\cM'|$, even after performing the above tweak.  
Indeed, that tweak only causes the time windows to be twice longer, and it does not change the fact that
once the computation of $\cM'$ finishes, the output matching is switched from the old matching $\cM$ to the new one $\cM'$ \emph{instantaneously}, which may require $|\cM| + |\cM'|$ replacements to the output matching at that time. 

\section{Proofs of Theorems \ref{main} and \ref{main2}} \label{proofmain}
This section is mostly devoted (see Sections 3.1 and 3.2) to the proof of Theorem~\ref{main}.
At the end of this section (Section 3.3) we sketch the adjustments needed for deriving the  result of Theorem \ref{main2},
whose proof follows along similar lines to those of Theorem \ref{main}.
\vspace{8pt}\\
\noindent{\bf 3.1~ A simple transformation in static graphs.~}
This section is devoted to the proof of Theorem \ref{th:MCM}, which provides the first step in the proof of Theorem \ref{main}. We remark that this theorem can be viewed as a ``warm up'' to Theorem \ref{th:main} for MWM, which is deferred to Section~\ref{sec:wema}, and is considerably more technically involved.

Let $\cM$ and $\cM'$ be two matchings for the same graph $G$.
Our goal is to gradually transform $\cM$ into (a possibly superset of) $\cM'$ via a sequence of constant-time operations to be described next, each making at most 3 changes to the matching,
 such that the matching obtained at any point throughout this   transformation process is a valid matching for $G$ of size at least $\min\{|\cM|,|\cM'|-1\}$.
It is technically convenient to denote by $\cM^*$ the \emph{transformed} matching,
which is initialized as $\cM$ at the outset, and being gradually transformed into $\cM'$; we refer to $\cM$ and $\cM'$ as the \emph{source} and \emph{target} matchings, respectively.
Each   operation starts by adding a single edge of $\cM' \setminus \cM^*$ to $\cM^*$ and then removing from $\cM^*$ the at most two edges incident on the newly added edge; thus at most 3 changes to the matching are made per operation.
It is instructive to assume that $|\cM'| > |\cM|$, as the   motivation for applying this transformation, which will become clear in Section 3.2, is to increase the matching size;
in this case the size $|\cM^*|$ of the transformed matching $\cM^*$ never goes below the size $|\cM|$ of the source matching $\cM$.
%

We say that an edge of $\cM' \setminus \cM^*$ that is incident on at most one edge of $\cM^*$ is \emph{good}, otherwise it is \emph{bad}, being incident on two edges of $\cM^*$.
Since $\cM^*$  has to be a valid matching throughout the transformation process, adding a bad edge to $\cM^*$ must trigger
the removal of two edges from $\cM^*$.
Thus if we keep adding bad edges to $\cM^*$, the size of $\cM^*$ may halve throughout the transformation process.
The following lemma shows that if all edges of $\cM' \setminus \cM^*$ are bad, the transformed matching $\cM^*$ is just as large as the target matching $\cM'$.
\begin{lemma} \label{badedge}
If all edges of $\cM' \setminus \cM^*$ are bad, then $|\cM^*| \ge |\cM'|$.
\end{lemma}
\begin{proof}
Consider a bipartite graph $L \cup R$, where each vertex in $L$ corresponds to an edge of $\cM' \setminus \cM^*$
and each vertex in $R$ corresponds to an edge of $\cM^* \setminus \cM'$, and there is an edge between a vertex in $L$ and a vertex in $R$ iff the corresponding matched edges share a common vertex in the original graph.
If all edges of $\cM' \setminus \cM^*$ are bad, then any edge of $\cM' \setminus \cM^*$ is incident on two edges of $\cM^*$, and since $\cM'$ is a valid matching,
those two edges cannot be in $\cM'$. In other words, the degree of each vertex in $L$ is exactly 2.
Also, the degree of each vertex in $R$ is at most 2, as $\cM'$ is a valid matching.
 It follows that $|R| \ge |L|$, or in other words $|\cM^* \setminus \cM'| \ge |\cM' \setminus \cM^*|$, yielding $|\cM^*| \ge |\cM'|$.
\QED
\end{proof}

The transformation process is carried out as follows.
At the outset we initialize $\cM^* = \cM$ and compute the sets $\cG$ and $\cB$ of good and bad edges in $\cM' \setminus \cM^* = \cM' \setminus \cM$ within time $O(|\cM| + |\cM'|)$ in the obvious way, and store them in doubly-linked lists. We keep mutual pointers between each edge of $\cM^*$ and its at most two incident edges
in the corresponding linked lists $\cG$ and $\cB$.
Then we perform a sequence of operations, where each operation starts by adding an edge of $\cM' \setminus \cM^*$ to $\cM^*$,
giving precedence to good edges (i.e., adding a bad edge to $\cM^*$ only when there are no good edges to add),
and then removing from $\cM^*$ the at most two edges incident on the newly added edge.
Following each such operation, we update the lists $\cG$ and $\cB$ of good and bad edges in $\cM' \setminus \cM^*$ within constant time in the obvious way.
This process is repeated until $\cM' \setminus \cM^* = \emptyset$,
at which stage we have $\cM^* \supseteq \cM'$.  
Note that the number of operations performed before emptying $\cM' \setminus \cM^*$ is bounded by $|\cM'|$, since each operation removes at least one edge from $\cM' \setminus \cM^*$.
It follows that the total runtime of the transformation process is bounded by $O(|\cM| + |\cM'|)$.

It is immediate that $\cM^*$ remains a valid matching throughout the transformation process,
as we pro-actively remove from it edges that share a common vertex with new edges added to it.
To complete the proof of Theorem  \ref{th:MCM} it remains to prove the following lemma.
%

\begin{lemma} \label{sizebound}
At any moment in time  we have $|\cM^*| \ge \min\{|\cM|,|\cM'|-1\}$.
\end{lemma}
\begin{proof}
Suppose for contradiction that the lemma does not hold, and consider the first time step $t^*$ throughout the transformation process in which $|\cM^*| < \min\{|\cM|,|\cM'|-1\}$.
Since initially $|\cM^*| = |\cM|$ and
as every addition of a good edge to $\cM^*$ triggers at most one edge removal from it,
 time step $t^*$ must occur after an addition of a bad edge.
Recall that a bad edge is added to $\cM^*$ only when there are no good edges to add.
Just before this addition we have $|\cM^*| \ge |\cM'|$ by Lemma \ref{badedge},
thus we have $|\cM^*| \ge |\cM'| - 1$ after adding that edge to $\cM^*$ and removing the two edges incident on it from there,
yielding a contradiction.
\QED
\end{proof}

\begin {remark}
\label{re:opt}
When $|\cM| < |\cM'|$, it is possible to gradually transform $\cM$ to $\cM'$ without ever being in deficit compared to the initial value of $\cM$, i.e., $|\cM^*| \ge |\cM|$ throughout the transformation process.
However, if $|\cM'| \le |\cM|$, this no longer holds true; refer to App.\ \ref{tightnessun} for more details.
\end {remark}
\noindent \vspace{-2pt} \\
{\bf 3.2~ The Fully Dynamic Setting.~}
In this section we provide the second step in the proof of Theorem \ref{main}, showing that the simple transformation process described in Section 3.1 for static graphs can be
generalized for the fully dynamic setting, thus completing the proof of Theorem \ref{main}.

Consider an arbitrary dynamic algorithm, Algorithm $\cA$, for maintaining a $\beta$-MCM with an update time of $T$, for any $\beta  \ge 1$ and $T$. The matching maintained by Algorithm $\cA$, denoted by $\cM^A_i$, for $i = 1,2,\ldots$, may change significantly following a single update step.
All that is guaranteed by Algorithm $\cA$ is that it can update the matching following every update step within a time bound of $T$, either in the worst-case sense or in the amortized sense, following which queries regarding the matching can be answered in (nearly) constant time.
Recall also that we assume that,
for any update step $i$, 
the matching $\cM^A_i$ provided by Algorithm $\cA$ at step $i$  can be output within time (nearly) linear in the matching size.

Our goal is to output a matching $\tilde \cM = \tilde \cM_i$, for $i = 1,2,\ldots$, possibly very different from $\cM^A = \cM^A_i$, which changes very slightly from one update step to the next.
To this end, the basic idea is to use the matching $\cM^A$  provided by Algorithm $\cA$ at a certain update step,
and then re-use it (gradually removing from it edges that get deleted from the graph) throughout a sufficiently long window of $\Theta(\eps \cdot |\cM^A|)$ consecutive update steps,
while gradually transforming it into a larger matching, provided again by Algorithm $\cA$ at some later step.

The \emph{gradual transformation process} is obtained by adapting the process described in Section 3.1 for static graphs to the fully dynamic setting.
Next, we describe this adaptation.
We   assume that $\beta = O(1)$;
the case of a general $\beta$ is addressed in Section 3.2.1.   

Consider the beginning of a new time window, at some update step $t$.  
Denote the matching provided by Algorithm $\cA$ at that stage by $\cM' = \cM^A_t$ and the matching output by our algorithm by $\cM = \tilde{\cM}_t$.
Recall that the entire matching $\cM' = \cM^A_t$ can be output in time (nearly) linear in its size, and we henceforth assume that $\cM'$ is given as a list of edges.
(For concreteness, we assume that the time needed for storing the edges of $\cM'$ in an appropriate list is $O(|\cM'|$.)
While $\cM'$ is guaranteed to provide a $\beta$-MCM at any update step, including $t$,
the approximation guarantee of $\cM$ may be worse. Nevertheless, we will show (Lemma \ref{complete}) that $\cM$ provides a
$(\beta(1+2\eps'))$-MCM for $G_t$. Under the assumption that $\beta = O(1)$, we thus have $|\cM| = O(|\cM'|)$.
The length of the time window is $W = \Theta(\eps \cdot |\cM|)$, i.e., it starts at update step $t$ and ends at update step $t' = t + W-1$.
During this time window, we gradually transform $\cM$ into (a possibly superset of) $\cM'$, using the transformation described in Section 3.1 for static graphs;
recall that the matching output throughout this transformation process is denoted by $\cM^*$.
We may assume that $|\cM|, |\cM'| = \Omega(1/\eps)$, where the constant hiding in the $\Omega$-notation is sufficiently large;
indeed, otherwise $|\cM| + |\cM'|  = O(1/\eps)$ and there is no need to apply the transformation process,
as the trivial worst-case recourse bound is $O(1/\eps)$.

We will show (Lemma \ref{complete}) that the output matching $\tilde \cM_i$ provides a $(\beta(1 + O(\eps))$-MCM at any update step $i$.
Two simple adjustments are needed for adapting the transformed matching $\cM^*$ of the static setting to the fully dynamic setting:
\begin{itemize}
\item
To achieve a low worst-case recourse bound and guarantee that the overhead in the update time (with respect to the original update time) is low in the worst-case,
we cannot carry out the entire computation at once (i.e. following a single update step), but should rather \emph{simulate it gradually} over the entire time window of the transformation process.
Specifically, recall that the transformation process for static graphs consists of two phases,
a preprocessing phase in which the matching $\cM' = \cM^A_t$ and the sets $\cG$ and $\cB$ of good and bad edges
 in $\cM' \setminus \cM$ are computed, and the actual transformation phase that transforms $\cM^*$, which is initialized as $\cM$, into (a possibly superset of) $\cM'$.
 Each of these phases requires time $O(|\cM| + |\cM'|) = O(|\cM|)$. The first phase does not make any replacements to $\cM^*$,
 whereas the second phase consists of a sequence of at most $|\cM'|$ constant-time operations, each of which may trigger a constant number of replacements to $\cM^*$.
The computation of the first phase is simulated
 in the first $W/2$ update steps of the window, performing $O(|\cM| + |\cM'|) / (W/2) = O(1/\eps)$ computation steps and zero replacements to $\cM^*$ following every update step.
The computation of the second phase is simulated in the second $W/2$ update steps of the window,
performing $O(|\cM| + |\cM'|) / (W/2) = O(1/\eps)$ computation steps and replacements to $\cM^*$ following every update step.
\item
Denote by $\cM^*_i$ the matching output at the $i$th update step by the resulting gradual transformation process, which simulates $O(1/\eps)$ computation steps and replacements to the output matching following every update step.
While $\cM^*_i$ is a valid matching for the (static) graph $G_t$ at the beginning of the time window, 
some of its edges may get deleted from the graph in subsequent update steps $i= t+1, t+2, \ldots, t'$.
Consequently, the matching that we shall output for graph $G_i$, denoted by $\tilde {\cM}_i$, is
the one obtained from $\cM^*_i$ by removing from it all edges that got deleted from the graph between steps $t$ and $i$.
\end{itemize}

Once the current time window terminates, a new time window starts, and the same transformation process is repeated, with $\tilde \cM_{t'}$ serving as $\cM$
and $\cM^A_{t'}$ serving as $\cM'$.
Since all time windows are handled in the same way,
it suffices to analyze the output matching of the current time window,
and this analysis would carry over to the entire update sequence.

It is immediate that the output matching $\tilde \cM_i$ is a valid matching for any $i = t, t+1, \ldots, t'$.
Moreover, since we make sure to simulate $O(1/\eps)$ computation steps and replacements following every update step,
the worst-case recourse bound of the resulting algorithm is bounded by $O(1/\eps)$ and the update time is bounded by $T + O(1/\eps)$,
where this time bound is worst-case/amortized if the time bound $T$ of Algorithm $\cA$ is worst-case/amortized.

It is left to bound the approximation guarantee of the output matching $\tilde \cM_i$.  
Recall that $W = \Theta(\eps \cdot |\cM|)$, and write $W = \eps' \cdot |\cM|$, with $\eps' = \Theta(\eps)$.
(We assume that $\eps$ is sufficiently small so that $\eps' \le 1/2$. We need this restriction on $\eps'$ to apply Lemma \ref{lazylemma2}.)
\begin{lemma} \label{complete}
$\tilde \cM_t$ and $\tilde \cM_{t'}$ provide a $(\beta(1+2\eps'))$-MCM for $G_t$ and $G_{t'}$, respectively.  
Moreover, $\tilde \cM_i$ provides  a $(\beta((1+2\eps')^2))$-MCM for $G_i$, for any $i = t, t+1, \ldots,t'$.
\end{lemma}
\begin{proof}
First, we bound the approximation guarantee of the matching $\tilde \cM_{t'}$,
which is obtained from $\cM^*_{t'}$ by removing from it all edges that got deleted from the graph throughout the time window.
By the description of the transformation process, $\cM^*_{t'}$ is a superset of $\cM'$,
hence $\tilde \cM_{t'}$ is a superset of the matching obtained from $\cM'$ by removing from it all edges that got deleted throughout the time window.
Since $\cM'$ is a $\beta$-MCM for $G_t$,
Lemma \ref{lazylemma2} implies that $\tilde \cM_{t'}$ is a $(\beta(1+2\eps'))$-MCM for $G_{t'}$.
More generally, this argument shows that the matching obtained at the end of any time window is a $(\beta(1+2\eps'))$-MCM for the graph at that step.

Next, we argue that the matching obtained at the start of any time window (as described above) is a $(\beta(1+2\eps'))$-MCM for
the graph at that step. This assertion is trivially true for the first time window, where both the matching and the graph are empty.
For any subsequent time window, this assertion follows from the fact that the matching at the start of a new time window is the one obtained at the end of the old time window, for which we have already shown that the required approximation guarantee holds.
It follows that $\tilde \cM_t = \cM$ is a $(\beta(1+2\eps'))$-MCM for $G_t$.

Finally, we bound the approximation guarantee of the output matching $\tilde \cM_i$ in the entire time window. (It suffices to consider the interior of the window.)
Lemma \ref{sizebound} implies that $|\cM^*_i| \ge \min\{|\cM|,|\cM'|-1\}$, for any $i = t, t+1, \ldots,t'$.
We argue that $\cM^*_i$ is a $(\beta(1+2\eps'))$-MCM for $G_t$.
If $|\cM^*_i| \ge |\cM|$, then this assertion follows from the fact that $\cM$ provides such an approximation guarantee.
We henceforth assume that $|\cM^*_i| \ge |\cM'|-1$. Recall that $|\cM'| = \Omega(1/\eps) = \Omega(1/\eps')$, where the constants hiding in the $\Omega$-notation are sufficiently large,
hence removing a single edge from $\cM'$ cannot hurt the approximation guarantee by more than an additive factor of, say $\eps'$, i.e., less than $\beta (2\eps')$.
Since $\cM'$ provides a $\beta$-MCM for $G_t$, it follows that $\cM^*_i$ is indeed a $(\beta(1+2\eps'))$-MCM for $G_t$,
which completes the proof of the above assertion.
Consequently, Lemma \ref{lazylemma2} implies that $\tilde \cM_{i}$, which is obtained from $\cM^*_i$
by removing from it all edges that got deleted from the graph between steps $t$ and $i$,
is a $(\beta((1+2\eps')^2))$-MCM for $G_{i}$.
\QED
\end{proof}
\vspace{7pt} \noindent 
{\bf 3.2.1~ A general approximation guarantee.~}
In this section we consider the case of a general approximation parameter $\beta \ge 1$.  
The bound on the approximation guarantee of the output matching provided by Lemma \ref{complete}, namely $(\beta((1+2\eps')^2))$, remains unchanged.
Recalling that $\eps' \le 1/2$, it follows that the size of $\cM'$ cannot be larger than that of $\cM$ by more than a factor of $(\beta((1+2\eps')^2)) \le 2\beta$.
Consequently, the number of computation steps and replacements performed per update step, namely, $O(|\cM| + |\cM'|) / (W/2)$, is no longer bounded by $O(1/\eps)$,
but rather by $O(\beta/\eps)$.
To achieve a bound of $O(1/\eps)$ for a general $\beta$, we shall use a matching $\cM''$ different from $\cM'$, which includes a possibly small fraction of the edges of $\cM'$.
Recall that we can output $\ell$ arbitrary edges of the matching $\cM' = \cM^A_t$ in time (nearly) linear in $\ell$, for any integer $\ell = 1,2,\ldots,|\cM'|$.
Let $\cM''$ be a matching that consists of (up to) $2|\cM|$ arbitrary edges of $\cM'$;
that is, if $|\cM'| > 2 |\cM|$, $\cM''$ consists of $2|\cM|$ arbitrary edges of $\cM'$, otherwise $\cM'' = \cM'$.
We argue that $\cM''$ is a $\beta$-MCM for $G_t$.
Indeed, if $|\cM'| > 2 |\cM|$ the approximation guarantee follows from the approximation guarantee of $\cM$ and the fact that $\cM''$ is twice larger than $\cM$,
whereas in the complementary case the approximation guarantee follows from that of $\cM'$.
In any case it is immediate that  $|\cM''| = O(|\cM|)$.
(For concreteness, we assume that the time needed for storing the edges of $\cM''$ in an appropriate list is $O(|\cM''|) = O(|\cM|)$.)
We may henceforth carry out the entire transformation process with $\cM''$ taking the role of $\cM'$, and in this way guarantee that
the number of computation steps and replacements to the output matching performed per update step is reduced from $O(\beta/\eps)$ to $O(1/\eps)$.
\noindent \vspace{8pt} \\
{\bf 3.3~ Proof of Theorem~\ref{main2}.~}
The proof of Theorem~\ref{main2} is very similar to the one of Theorem~\ref{main}. Specifically, we derive  Theorem~\ref{main2} by making a couple of simple adjustments to the proof of  Theorem~\ref{main} given above, which we sketch next.
First, instead of using the transformation of Theorem~\ref{th:MCM}, we use the one of Theorem~\ref{th:main}, whose proof appears in Section \ref{sec:wema}.
Second, the \emph{stability} property of unweighted matchings used in the proof of Theorem~\ref{main} is that   the maximum matching size changes by at most 1 following each update step.
This stability property enables us in the proof of Theorem \ref{main} to consider a time window of $W = \Theta(\eps \cdot |\cM|)$ update steps, so that any $\beta$-MCM computed at the beginning of the window will  provide (after removing from it all the edges that get deleted from the graph) a $(\beta(1+\eps))$-MCM throughout the entire window, for any $\beta \ge 1$.
It is easy to see that this stability property   generalizes for weighted matchings, where the maximum matching weight may change by an additive factor of at most $\psi$. (Recall that the aspect ratio of the dynamic graph is always bounded by $\psi$.)
In order to obtain a $(\beta(1+\eps))$-MWM throughout the entire time window, it suffices to consider a time window of $W' = W'_\psi = W / \psi = \Theta(\eps \cdot |\cM| / \psi)$, i.e.,
 a time window  shorter than that used for unweighted matchings by a factor of $\psi$,
and as a result the update time of the resulting algorithm will grow from $T + O(1/\eps)$ to $T + O(\psi/\eps)$ and the worst-case recourse bound will grow from $O(1/\eps)$ to $O(\psi/\eps)$.

\section{Optimality of the Recourse Bound} \label{recourse}
In this section we show that an approximation guarantee of $(1+\eps)$ requires a recourse bound of $\Omega(1/\eps)$,
even in the amortized sense and even in the incremental (insertion only) and decremental (deletion only) settings.
We only consider edge updates, but the argument extends seamlessly to vertex updates.
This lower bound of $\Omega(1/\eps)$ on the recourse bound does not depend on the update time of the algorithm in any way.
Let us fix $\eps$ to be any parameter satisfying $\eps = \Omega(1/n), \eps \ll 1$, where $n$ is the (fixed) number of vertices.


Consider a simple path $P_\ell = (v_1,v_2,\ldots,v_{2\ell})$ of length $2\ell-1$, for an  integer $\ell = c(1/\eps)$ such that $\ell \ge 1$ and $c$ is a sufficiently small constant. (Thus $P_\ell$ spans at least two but no more than $n$ vertices.) There is a single maximum matching $\cM^{OPT}_\ell$ for $P_\ell$, of size $\ell$, which is also the only $(1+\eps)$-MCM for $P_\ell$.
After adding the two edges $(v_0,v_1)$ and $(v_{2\ell},v_{2\ell+1})$ to $P_\ell$, the maximum matching $\cM^{OPT}_\ell$ for the old path $P_\ell$
does not provide a $(1+\eps)$-MCM for the new path, $(v_0,v_1,\ldots,v_{2\ell+1})$, which we may rewrite as $P_{\ell + 1} = (v_1,v_2,\ldots, v_{2(\ell+1)})$.
The only way to restore a $(1+\eps)$-approximation guarantee is by removing all $\ell$ edges of $\cM^{OPT}_\ell$ and adding the remaining $\ell + 1$ edges instead,
which yields $\cM^{OPT}_{\ell+1}$.
One may carry out this argument repeatedly
until the length of the path reaches, say, $4\ell -1$.  
The amortized number of replacements to the matching per update step throughout this process is $\Omega(1/\eps)$.
Moreover, the same amortized bound, up to a small constant factor, holds if we start from an empty path instead of a path of length $2\ell-1$.  
We then delete all $4\ell - 1$ edges of the final path and start again from scratch, which may reduce the amortized bound by another small constant.
In this way we get an amortized recourse bound of $\Omega(1/\eps)$ for the fully dynamic setting.

To adapt this lower bound to the incremental setting, we construct $n' = \Theta(\eps \cdot n)$ vertex-disjoint copies $P^1, P^2,\ldots,P^{n'}$
of the aforementioned incremental path, one after another, in the following way.
Consider the $i$th copy $P^i$, from the moment its length becomes $2\ell - 1$ and until it reaches $4\ell - 1$.
If at any moment  during this gradual construction of $P^i$, the matching restricted to $P^i$ is not the (only) maximum matching for $P^i$,
we \emph{halt} the construction of $P^i$ and move on to constructing the $(i+1)$th copy $P^{i+1}$, and then subsequent copies, in the same way.
A copy whose construction started but was halted is called \emph{incomplete}; otherwise it is \emph{complete}.
(There are also \emph{empty} copies, whose construction has not started yet.)
For any  incomplete copy $P^j$, the matching restricted to it
is not the   maximum matching for $P^j$,
hence its approximation guarantee is worse than $1+\eps$; more precisely, the approximation guarantee provided by any matching other than the maximum matching for $P^j$
is at least $1 + c' \cdot \eps$, for a constant $c'$ that can be made as large as we want by decreasing the aforementioned constant $c$, or equivalently, $\ell$.
(Recall that $\ell = c(1/\eps)$.)
If the matching restricted to $P^j$ is changed to the maximum matching for $P^j$
at some later moment in time,  
we return to that incomplete copy and resume its construction from where we left off, thereby \emph{temporarily suspending} the construction of some other copy $P_{j'}$.
The construction of $P^j$ may get halted again, in which case we return to handling the temporarily suspended copy $P_{j'}$, otherwise we return to handling
$P_{j'}$ only after the construction of $P^j$ is complete, and so forth.
In this way we maintain the invariant that the approximation guarantee of the matching restricted to any incomplete copy (whose construction is not temporarily suspended) is at least $1 + c' \cdot \eps$, for a sufficiently large constant $c'$.
While incomplete copies may get completed later on, a complete copy remains complete throughout the entire update sequence.
At the end of the update sequence no copy is empty or temporarily suspended, i.e., any copy at the end of the update sequence is either incomplete or complete.
The above argument implies that any complete copy has an amortized recourse bound of $\Omega(1/\eps)$, over the update steps restricted to that copy.
Observe also that at least a constant fraction of the $n'$ copies must be complete at the end of the update sequence, otherwise the entire matching cannot provide a $(1+\eps)$-MCM for the entire graph, i.e., the graph obtained from the union of these $n'$ copies.
It follows that the amortized recourse bound over the entire update sequence is $\Omega(1/\eps)$.

The lower bound for the incremental setting can be extended to the decremental setting using a symmetric argument to the one given above.


\clearpage
\bibliographystyle{latex8}
\bibliography{randomMMbibfile}

\begin{thebibliography}{10}\setlength{\itemsep}{-1ex}\small

\bibitem{ADKKP16}
I.~Abraham, D.~Durfee, I.~Koutis, S.~Krinninger, and R.~Peng.
\newblock On fully dynamic graph sparsifiers.
\newblock In {\em Proc. of 57th FOCS}, pages 335--344, 2016.

\bibitem{ADJ18}
S.~Angelopoulos, C.~D{\"{u}}rr, and S.~Jin.
\newblock Online maximum matching with recourse.
\newblock In {\em 43rd International Symposium on Mathematical Foundations of
  Computer Science, {MFCS} 2018, August 27-31, 2018, Liverpool, {UK}}, pages
  8:1--8:15, 2018.

\bibitem{ACCSW18}
M.~Arar, S.~Chechik, S.~Cohen, C.~Stein, and D.~Wajc.
\newblock Dynamic matching: Reducing integral algorithms to
  approximately-maximal fractional algorithms.
\newblock In {\em Proc. 45th ICALP}, pages 7:1--7:16, 2018.

\bibitem{ABBMS17}
S.~Assadi, M.~Bateni, A.~Bernstein, V.~S. Mirrokni, and C.~Stein.
\newblock Coresets meet {EDCS:} algorithms for matching and vertex cover on
  massive graphs.
\newblock {\em CoRR}, abs/1711.03076, 2017.

\bibitem{AK17}
S.~Assadi and S.~Khanna.
\newblock Randomized composable coresets for matching and vertex cover.
\newblock In {\em Proceedings of the 29th {ACM} Symposium on Parallelism in
  Algorithms and Architectures, {SPAA} 2017, Washington DC, USA, July 24-26,
  2017}, pages 3--12, 2017.

\bibitem{AOSS18}
S.~Assadi, K.~Onak, B.~Schieber, and S.~Solomon.
\newblock Fully dynamic maximal independent set with sublinear update time.
\newblock In {\em Proc. 50th STOC}, 2018 (to appear).

\bibitem{ausiello2011complexity}
G.~Ausiello, V.~Bonifaci, and B.~Escoffier.
\newblock Complexity and approximation in reoptimization.
\newblock In {\em Computability in Context: Computation and Logic in the Real
  World}, pages 101--129. World Scientific, 2011.

\bibitem{AusielloEMP09}
G.~Ausiello, B.~Escoffier, J.~Monnot, and V.~T. Paschos.
\newblock Reoptimization of minimum and maximum traveling salesman's tours.
\newblock {\em J. Discrete Algorithms}, 7(4):453--463, 2009.

\bibitem{BGKPSS15}
N.~Bansal, A.~Gupta, R.~Krishnaswamy, K.~Pruhs, K.~Schewior, and C.~Stein.
\newblock A 2-competitive algorithm for online convex optimization with
  switching costs.
\newblock In {\em Proc. of APPROX-RANDOM}, pages 96--109, 2015.

\bibitem{BGS11}
S.~Baswana, M.~Gupta, and S.~Sen.
\newblock Fully dynamic maximal matching in ${O}(\log n)$ update time.
\newblock In {\em Proc. of 52nd FOCS}, pages 383--392, 2011 (see also {\em
  SICOMP'15} version, and subsequent erratum).

\bibitem{bender2015reallocation}
M.~A. Bender, M.~Farach-Colton, S.~P. Fekete, J.~T. Fineman, and S.~Gilbert.
\newblock Reallocation problems in scheduling.
\newblock {\em Algorithmica}, 73(2):389--409, 2015.

\bibitem{bender2017cost}
M.~A. Bender, M.~Farach-Colton, S.~P. Fekete, J.~T. Fineman, and S.~Gilbert.
\newblock Cost-oblivious storage reallocation.
\newblock {\em ACM Transactions on Algorithms (TALG)}, 13(3):38, 2017.

\bibitem{BFH19}
A.~Bernstein, S.~Forster, and M.~Henzinger.
\newblock A deamortization approach for dynamic spanner and dynamic maximal
  matching.
\newblock In {\em Proceedings of the Thirtieth Annual {ACM-SIAM} Symposium on
  Discrete Algorithms, {SODA} 2019, San Diego, California, USA, January 6-9,
  2019}, pages 1899--1918, 2019.

\bibitem{bernstein2019deamortization}
A.~Bernstein, S.~Forster, and M.~Henzinger.
\newblock A deamortization approach for dynamic spanner and dynamic maximal
  matching.
\newblock In {\em Proceedings of the Thirtieth Annual ACM-SIAM Symposium on
  Discrete Algorithms}, pages 1899--1918. SIAM, 2019.

\bibitem{BHR17}
A.~Bernstein, J.~Holm, and E.~Rotenberg.
\newblock Online bipartite matching with amortized
  {\textdollar}o({\textbackslash}log{\^{}}2 n){\textdollar} replacements.
\newblock In {\em Proc. of 28th SODA}, pages 692--711, 2018.

\bibitem{BKPPS17}
A.~Bernstein, T.~Kopelowitz, S.~Pettie, E.~Porat, and C.~Stein.
\newblock Simultaneously load balancing for every p-norm, with reassignments.
\newblock In {\em Proc. 8th ITCS}, pages 51:1--51:14, 2017.

\bibitem{BS15}
A.~Bernstein and C.~Stein.
\newblock Fully dynamic matching in bipartite graphs.
\newblock In {\em Proc. 42nd ICALP}, pages 167--179, 2015.

\bibitem{BS16}
A.~Bernstein and C.~Stein.
\newblock Faster fully dynamic matchings with small approximation ratios.
\newblock In {\em Proc. of 26th SODA}, pages 692--711, 2016.

\bibitem{BHN16}
S.~Bhattacharya, M.~Henzinger, and D.~Nanongkai.
\newblock New deterministic approximation algorithms for fully dynamic
  matching.
\newblock In {\em Proc. 48th STOC}, pages 398--411, 2016.

\bibitem{BHN17}
S.~Bhattacharya, M.~Henzinger, and D.~Nanongkai.
\newblock Fully dynamic maximum matching and vertex cover in $o(\log^3 n)$
  worst case update time.
\newblock In {\em Proc. of 28th SODA}, pages 470--489, 2017.

\bibitem{Bilo18}
D.~Bil{\`{o}}.
\newblock New algorithms for steiner tree reoptimization.
\newblock In {\em 45th International Colloquium on Automata, Languages, and
  Programming, {ICALP} 2018, July 9-13, 2018, Prague, Czech Republic}, pages
  19:1--19:14, 2018.

\bibitem{BiloBKKMSZ11}
D.~Bil{\`{o}}, H.~B{\"{o}}ckenhauer, D.~Komm, R.~Kr{\'{a}}lovic,
  T.~M{\"{o}}mke, S.~Seibert, and A.~Zych.
\newblock Reoptimization of the shortest common superstring problem.
\newblock {\em Algorithmica}, 61(2):227--251, 2011.

\bibitem{DBLP:conf/mfcs/BonamyBHIKMMW19}
M.~Bonamy, N.~Bousquet, M.~Heinrich, T.~Ito, Y.~Kobayashi, A.~Mary,
  M.~M{\"{u}}hlenthaler, and K.~Wasa.
\newblock The perfect matching reconfiguration problem.
\newblock In P.~Rossmanith, P.~Heggernes, and J.~Katoen, editors, {\em 44th
  International Symposium on Mathematical Foundations of Computer Science,
  {MFCS} 2019, August 26-30, 2019, Aachen, Germany}, volume 138 of {\em
  LIPIcs}, pages 80:1--80:14. Schloss Dagstuhl - Leibniz-Zentrum f{\"{u}}r
  Informatik, 2019.

\bibitem{bonsma2009finding}
P.~Bonsma and L.~Cereceda.
\newblock Finding paths between graph colourings: Pspace-completeness and
  superpolynomial distances.
\newblock {\em Theoretical Computer Science}, 410(50):5215--5226, 2009 (See
  also Proc. of MFCS'07).

\bibitem{BoriaP10}
N.~Boria and V.~T. Paschos.
\newblock Fast reoptimization for the minimum spanning tree problem.
\newblock {\em J. Discrete Algorithms}, 8(3):296--310, 2010.

\bibitem{BLSZ14}
B.~Bosek, D.~Leniowski, P.~Sankowski, and A.~Zych.
\newblock Online bipartite matching in offline time.
\newblock In {\em Proc. 55th FOCS}, pages 384--393, 2014.

\bibitem{BLSZ15}
B.~Bosek, D.~Leniowski, P.~Sankowski, and A.~Zych.
\newblock Shortest augmenting paths for online matchings on trees.
\newblock In {\em Proc. of 13th WAOA}, pages 59--71, 2015.

\bibitem{BLZS17}
B.~Bosek, D.~Leniowski, P.~Sankowski, and A.~Zych-Pawlewicz.
\newblock A tight bound for shortest augmenting paths on trees.
\newblock In {\em Proc. 13th LATIN}, pages 201--216, 2018.

\bibitem{DBLP:conf/wg/BousquetHIM19}
N.~Bousquet, T.~Hatanaka, T.~Ito, and M.~M{\"{u}}hlenthaler.
\newblock Shortest reconfiguration of matchings.
\newblock In I.~Sau and D.~M. Thilikos, editors, {\em Graph-Theoretic Concepts
  in Computer Science - 45th International Workshop, {WG} 2019, Vall de
  N{\'{u}}ria, Spain, June 19-21, 2019, Revised Papers}, volume 11789 of {\em
  Lecture Notes in Computer Science}, pages 162--174. Springer, 2019.

\bibitem{CHK16}
K.~Censor{-}Hillel, E.~Haramaty, and Z.~S. Karnin.
\newblock Optimal dynamic distributed {MIS}.
\newblock In {\em Proceedings of the 2016 {ACM} Symposium on Principles of
  Distributed Computing, {PODC} 2016, Chicago, IL, USA, July 25-28, 2016},
  pages 217--226, 2016.

\bibitem{CS18}
M.~Charikar and S.~Solomon.
\newblock Fully dynamic almost-maximal matching: Breaking the polynomial
  worst-case time barrier.
\newblock In {\em Proc. 45th ICALP}, pages 33:1--33:14, 2018.

\bibitem{CDKL09}
K.~Chaudhuri, C.~Daskalakis, R.~D. Kleinberg, and H.~Lin.
\newblock Online bipartite perfect matching with augmentations.
\newblock In {\em Proc. of 28th INFOCOM}, pages 1044--1052, 2009.

\bibitem{gopalan2009connectivity}
P.~Gopalan, P.~G. Kolaitis, E.~Maneva, and C.~H. Papadimitriou.
\newblock The connectivity of boolean satisfiability: computational and
  structural dichotomies.
\newblock {\em SIAM Journal on Computing}, 38(6):2330--2355, 2009 (See also
  proc.\ of ICALP'06).

\bibitem{GLSSS19}
F.~Grandoni, S.~Leonardi, P.~Sankowski, C.~Schwiegelshohn, and S.~Solomon.
\newblock {(1} + {\(\epsilon\)})-approximate incremental matching in constant
  deterministic amortized time.
\newblock In {\em Proceedings of the Thirtieth Annual {ACM-SIAM} Symposium on
  Discrete Algorithms, {SODA} 2019, San Diego, California, USA, January 6-9,
  2019}, pages 1886--1898, 2019.

\bibitem{GKKV95}
E.~F. Grove, M.~Kao, P.~Krishnan, and J.~S. Vitter.
\newblock Online perfect matching and mobile computing.
\newblock In {\em Proc. of 45th Wads}, pages 194--205, 1995.

\bibitem{GGK13}
A.~Gu, A.~Gupta, and A.~Kumar.
\newblock The power of deferral: maintaining a constant-competitive steiner
  tree online.
\newblock In {\em Symposium on Theory of Computing Conference, STOC'13, Palo
  Alto, CA, USA, June 1-4, 2013}, pages 525--534, 2013.

\bibitem{GKKP17}
A.~Gupta, R.~Krishnaswamy, A.~Kumar, and D.~Panigrahi.
\newblock Online and dynamic algorithms for set cover.
\newblock In {\em Proc. 49th STOC}, pages 537--550, 2017.

\bibitem{GK14}
A.~Gupta and A.~Kumar.
\newblock Online steiner tree with deletions.
\newblock In {\em Proceedings of the Twenty-Fifth Annual {ACM-SIAM} Symposium
  on Discrete Algorithms, {SODA} 2014, Portland, Oregon, USA, January 5-7,
  2014}, pages 455--467, 2014.

\bibitem{GKS14}
A.~Gupta, A.~Kumar, and C.~Stein.
\newblock Maintaining assignments online: Matching, scheduling, and flows.
\newblock In {\em Proc. 25th SODA}, pages 468--479, 2014.

\bibitem{gupta2019complexity}
M.~Gupta, H.~Kumar, and N.~Misra.
\newblock On the complexity of optimal matching reconfiguration.
\newblock In {\em International Conference on Current Trends in Theory and
  Practice of Informatics}, pages 221--233. Springer, 2019.

\bibitem{GP13}
M.~Gupta and R.~Peng.
\newblock Fully dynamic $(1+\epsilon)$-approximate matchings.
\newblock In {\em 54th FOCS}, pages 548--557, 2013.

\bibitem{DBLP:conf/icalp/HearnD02}
R.~A. Hearn and E.~D. Demaine.
\newblock The nondeterministic constraint logic model of computation:
  Reductions and applications.
\newblock In P.~Widmayer, F.~T. Ruiz, R.~M. Bueno, M.~Hennessy, S.~J.
  Eidenbenz, and R.~Conejo, editors, {\em Automata, Languages and Programming,
  29th International Colloquium, {ICALP} 2002, Malaga, Spain, July 8-13, 2002,
  Proceedings}, volume 2380 of {\em Lecture Notes in Computer Science}, pages
  401--413. Springer, 2002.

\bibitem{HK73}
J.~E. Hopcroft and R.~M. Karp.
\newblock An n$^{\mbox{5/2}}$ algorithm for maximum matchings in bipartite
  graphs.
\newblock {\em SIAM J. Comput.}, 2(4):225--231, 1973.

\bibitem{DBLP:journals/tcs/ItoDHPSUU11}
T.~Ito, E.~D. Demaine, N.~J.~A. Harvey, C.~H. Papadimitriou, M.~Sideri,
  R.~Uehara, and Y.~Uno.
\newblock On the complexity of reconfiguration problems.
\newblock {\em Theor. Comput. Sci.}, 412(12-14):1054--1065, 2011 (See also
  proc.\ of ISAAC'08).

\bibitem{DBLP:conf/esa/ItoKK0O19}
T.~Ito, N.~Kakimura, N.~Kamiyama, Y.~Kobayashi, and Y.~Okamoto.
\newblock Shortest reconfiguration of perfect matchings via alternating cycles.
\newblock In M.~A. Bender, O.~Svensson, and G.~Herman, editors, {\em 27th
  Annual European Symposium on Algorithms, {ESA} 2019, September 9-11, 2019,
  Munich/Garching, Germany}, volume 144 of {\em LIPIcs}, pages 61:1--61:15.
  Schloss Dagstuhl - Leibniz-Zentrum f{\"{u}}r Informatik, 2019.

\bibitem{DBLP:journals/jco/ItoKKKO19}
T.~Ito, N.~Kakimura, N.~Kamiyama, Y.~Kobayashi, and Y.~Okamoto.
\newblock Reconfiguration of maximum-weight b-matchings in a graph.
\newblock {\em J. Comb. Optim.}, 37(2):454--464, 2019 (Also in Proc. of
  COCOON'17).

\bibitem{DBLP:journals/tcs/KaminskiMM12}
M.~Kaminski, P.~Medvedev, and M.~Milanic.
\newblock Complexity of independent set reconfigurability problems.
\newblock {\em Theor. Comput. Sci.}, 439:9--15, 2012.

\bibitem{KKM13}
B.~M. Kapron, V.~King, and B.~Mountjoy.
\newblock Dynamic graph connectivity in polylogarithmic worst case time.
\newblock In {\em Proc. of 24th SODA}, pages 1131--1142, 2013.

\bibitem{KR08}
S.~Khot and O.~Regev.
\newblock Vertex cover might be hard to approximate to within 2-epsilon.
\newblock {\em J. Comput. Syst. Sci.}, 74(3):335--349, 2008.

\bibitem{MSV18}
J.~Matuschke, U.~Schmidt{-}Kraepelin, and J.~Verschae.
\newblock Maintaining perfect matchings at low cost.
\newblock {\em CoRR}, abs/1811.10580, 2018.

\bibitem{MSVW16}
N.~Megow, M.~Skutella, J.~Verschae, and A.~Wiese.
\newblock The power of recourse for online {MST} and {TSP}.
\newblock {\em {SIAM} J. Comput.}, 45(3):859--880, 2016.

\bibitem{MV80}
S.~Micali and V.~V. Vazirani.
\newblock An ${O}(\sqrt{|{V}|} |{E}|)$ algorithm for finding maximum matching
  in general graphs.
\newblock In {\em Proc. 21st FOCS}, pages 17--27, 1980.

\bibitem{MZ15}
V.~S. Mirrokni and M.~Zadimoghaddam.
\newblock Randomized composable core-sets for distributed submodular
  maximization.
\newblock In {\em Proceedings of the Forty-Seventh Annual {ACM} on Symposium on
  Theory of Computing, {STOC} 2015, Portland, OR, USA, June 14-17, 2015}, pages
  153--162, 2015.

\bibitem{NS13}
O.~Neiman and S.~Solomon.
\newblock Simple deterministic algorithms for fully dynamic maximal matching.
\newblock In {\em Proc. 45th STOC}, pages 745--754, 2013.

\bibitem{DBLP:journals/algorithms/Nishimura18}
N.~Nishimura.
\newblock Introduction to reconfiguration.
\newblock {\em Algorithms}, 11(4):52, 2018.

\bibitem{OR10}
K.~Onak and R.~Rubinfeld.
\newblock Maintaining a large matching and a small vertex cover.
\newblock In {\em Proc. of 42nd STOC}, pages 457--464, 2010.

\bibitem{PS16}
D.~Peleg and S.~Solomon.
\newblock Dynamic $(1 + \epsilon)$-approximate matchings: A density-sensitive
  approach.
\newblock In {\em Proc. of 26th SODA (to appear)}, 2016.

\bibitem{SSTT18}
B.~Schieber, H.~Shachnai, G.~Tamir, and T.~Tamir.
\newblock A theory and algorithms for combinatorial reoptimization.
\newblock {\em Algorithmica}, 80(2):576--607, 2018.

\bibitem{Sol16}
S.~Solomon.
\newblock Fully dynamic maximal matching in constant update time.
\newblock In {\em Proc. 57th FOCS}, pages 325--334, 2016.

\bibitem{thiongane2006}
B.~Thiongane, A.~Nagih, and G.~Plateau.
\newblock Lagrangean heuristics combined with reoptimization for the 0--1
  bidimensional knapsack problem.
\newblock {\em Discrete applied mathematics}, 154(15):2200--2211, 2006.

\bibitem{DBLP:books/cu/p/Heuvel13}
J.~van~den Heuvel.
\newblock The complexity of change.
\newblock In S.~R. Blackburn, S.~Gerke, and M.~Wildon, editors, {\em Surveys in
  Combinatorics 2013}, volume 409 of {\em London Mathematical Society Lecture
  Note Series}, pages 127--160. Cambridge University Press, 2013.

\bibitem{Vaz12}
V.~V. Vazirani.
\newblock An improved definition of blossoms and a simpler proof of the {MV}
  matching algorithm.
\newblock {\em CoRR}, abs/1210.4594, 2012.

\bibitem{DBLP:journals/corr/abs-1911-05545}
D.~Wajc.
\newblock Rounding dynamic matchings against an adaptive adversary.
\newblock {\em CoRR}, abs/1911.05545, 2019.

\end{thebibliography}

\clearpage
\appendix
\centerline{\LARGE\bf Appendix}

\section{Scenarios with high recourse bounds} \label{12}
We briefly discuss some scenarios where 
high recourse bounds may naturally arise. In all such scenarios our reductions (Theorems \ref{main} and \ref{main2}) can come into play to achieve low worst-case recourse bounds;
for clarity we focus in this discussion, sometimes implicitly, on large (unweighted) matching, but the entire discussion carries over with very minor changes to
the generalized setting of weighted matchings.

Appendix A.3 demonstrates that, although we \emph{may not care at all about recourse bounds},
maintaining a large (weight) matching with a low update time requires in some cases the use of a dynamic matching algorithm with a low recourse bound;
this is another situation where our reductions can come into play, but more than that, we believe that it provides an additional strong motivation for our reductions.
\noindent \vspace{10pt} \\
{\bf A.1~ Randomized algorithms} \label{121}
\noindent \vspace{6pt} \\
{\bf Multiple matchings.~}
Given a randomized algorithm for maintaining a large matching in a dynamic graph, it may be advantageous to run multiple instances of the algorithm (say $\polylog(n)$),
since this may increase the chances that at least one of those instances provides a large matching with high probability (w.h.p.) at any point in time.
Notice, however, that it is not the same matching that is guaranteed to be large throughout the entire update sequence, hence
the ultimate algorithm (or data structure), which outputs the largest among the $\polylog(n)$ matchings,
may need to switch between a pool of possibly very different matchings when going from one update step to the next.
Thus even if the recourse bound of the given randomized algorithm is low, and so each of the maintained matchings changes gradually over time,
we do not get any nontrivial recourse bound for the ultimate algorithm.
\noindent \vspace{6pt} \\
{\bf Large matchings.~}
Sometimes the approximation guarantee of the given randomized algorithm holds w.h.p.\ only when the matching is sufficiently large.
This is the case with the algorithm of \cite{CS18} that achieves $\polylog(n)$ worst-case update time,
where the approximation guarantee of $2+\eps$ holds w.h.p.\ only when the size of the matching is $\Omega(\log^5 n / \eps^4)$.
To perform efficiently, \cite{CS18} also maintains a matching that is guaranteed to be maximal (and thus provide a 2-MCM)
when the maximum matching size is smaller than $\delta = O(\log^5 n / \eps^4)$, via a deterministic procedure with a worst-case update time of $O(\delta)$.
The ultimate algorithm of \cite{CS18} switches between the matching given by the randomized algorithm and that by the deterministic procedure, taking the larger of the two.
Thus even if the recourse bounds of both the randomized algorithm and the deterministic procedure are low,
the worst-case recourse bound of the ultimate algorithm, which might be of the order of the ``large matching'' threshold, could be very high.
(The large matching threshold is the threshold on the matching size above which a high probability bound on the approximation guarantee holds.)
In   \cite{CS18} the large matching threshold is $\delta =  O(\log^5 n / \eps^4)$, so the recourse bound is reasonably low.
(This is not the bottleneck for the recourse bound of \cite{CS18}, as discussed next.)
In general, however, 
the large matching threshold may be significantly higher than $\polylog(n)$.
\noindent \vspace{6pt} \\
{\bf Long update sequences.~}
For the probabilistic guarantees of a randomized dynamic algorithm to hold w.h.p., the update sequence must be of bounded length.
In particular, polylogarithmic guarantees on the update time usually require that the length of the update sequence will be polynomially bounded.
This is the case with numerous dynamic graph algorithms also outside the scope of graph matchings (cf.\ \cite{KKM13,ADKKP16}),
and the basic idea is to partition the update sequence into sub-sequences of polynomial length each and to run a fresh instance of the dynamic algorithm
in each sub-sequence.
In the context of matchings, the algorithm of \cite{CS18} uses this approach. 
Notice, however, that an arbitrary sub-sequence (other than the first) does not start from an empty graph.
Hence, for the ultimate algorithm of \cite{CS18} to provide a low worst-case update time, it has to gradually construct the graph at the beginning of each sub-sequence from scratch and maintain for it a new gradually growing matching, while re-using the ``old'' matching used for the previous sub-sequence throughout this gradual process. Once the gradually constructed graph coincides with the true graph,
the ultimate algorithm switches from the old matching to the new one. (See \cite{CS18} for further details.)
While this approach guarantees that the worst-case update time of the algorithm is in check, it does not provide any nontrivial worst-case recourse bound.
\noindent \vspace{10pt} \\
{\bf A.2~~ From amortized to worst-case.~}
There are techniques 
for transforming algorithms with low amortized bounds into algorithms with similar worst-case bounds.
For approximate matchings, such a technique was first presented in \cite{GP13}.
Alas, the transformed algorithms do not achieve any nontrivial worst-case recourse bound;
see Section \ref{changes} for details.  
\noindent \vspace{10pt} \\
{\bf A.3~ When low update time requires low recourse bound.~} \label{123}
When a dynamic matching algorithm is used as a black-box subroutine inside a larger data structure or algorithm,
a low recourse bound of the algorithm used as a subroutine is needed for achieving a low update time for the larger algorithm.
We next consider a natural question motivating this need; one may refer to \cite{BS16,ADKKP16} for additional motivation.
\begin{Question} \label{q}
Given $k$ dynamic matchings of a dynamic graph $G$, whose union is guaranteed to contain a large matching for $G$
at any  time, for an arbitrary parameter $k$, can we combine those $k$ matchings into a large dynamic matching for $G$ \emph{efficiently}?
\end{Question}

This question may arise when there are physical limitations, such as memory constraints, e.g., as captured by  MapReduce-style computation,
where the edges of the graph are partitioned into $k$ parties.
More specifically, consider a fully dynamic graph $G$ of huge scale, for which we want to maintain a large matching with low update time.
The edges of the graph are dynamically partitioned into $k$ parties due to memory constraints, each capable of maintaining a large matching for the graph induced by its own edges with low update time, and the only guarantee on those $k$ dynamically changing matchings is the following global one: The \emph{union} of the $k$ matchings at any point in time \emph{contains} a large matching for the entire dynamic graph $G$.
(E.g., if we maintain at each update step the invariant that the edges of $G$ are partitioned across the $k$ parties uniformly at random, such a global guarantee can be provided via the framework of \emph{composable randomized coresets}   \cite{MZ15,AK17,ABBMS17}.)

This question may also arise when the input data set is noisy.  
Coping with noisy input usually requires   \emph{randomization}, which may lead to high recourse bounds as discussed in Appendix A.1. 
Let us revisit the scenario where we run multiple instances of a randomized dynamic algorithm with low update time;  
denote the number of such instances by $k$.  
If the input is noisy, we may not be able to guarantee that at least one of the $k$ maintained matchings
is large w.h.p.\ at any point in time, as suggested in Appendix A.1. A weaker, more reasonable assumption is that the {union} of those $k$ matchings contains a large matching.

The key observation is that it is insufficient to maintain each of the $k$ matchings with low update time, even in the worst-case, as each such matching may change significantly following a single update step, thereby changing significantly the union of those matchings. ``Feeding'' this union to \emph{any} dynamic matching algorithm
would result with poor update time bounds, even in the amortized sense. Consequently, to resolve Question \ref{q}, each   of the $k$ maintained matchings must change \emph{gradually} over time, or in other words, the underlying algorithm(s) needed for maintaining those matchings  
should guarantee a low recourse bound. A low amortized/worst-case recourse bound of the underlying algorithm(s) translates into a low amortized/worst-case update time of the ultimate algorithm, provided of course that the underlying algorithm(s) for maintaining those $k$ matchings, as well as the dynamic matching algorithm to which their union is fed,
all achieve a low amortized/worst-case update time.

\section{Optimality of our Transformations} \label{tightness}

\subsection{Unweighted matchings} \label{tightnessun}
In the unweighted case, when $|\cM| < |\cM'|$, Theorem \ref{th:MCM} states that $\cM$ can gradually transform into $\cM'$ without ever being in deficit compared to the initial value of $\cM$, i.e., $|\cM^*| \ge |\cM|$ throughout the entire transformation process.
If $|\cM'| \le |\cM|$, however, this no longer holds; in this case the theorem states that we'll reach a deficit of at most 1 unit.
To see that this bound is tight,
consider the case when $|\cM| = |\cM'|$ and $H = \cM \oplus \cM'$ is a simple alternating cycle that consists of all edges in $\cM$ and $\cM'$, and thus  of length $2|\cM|$. 
Throughout any transformation process and until handling the last edge of the cycle, 
it must be that $|\cM^*| < |\cM|$ if $\Delta < 2|\cM|$. 
\vspace{5pt}
\\
\noindent{\bf Remark.} In fact, the same situation will occur if $\Delta = 2$.
In the particular case of $\Delta = 1$,   we'll be in deficit of up to 2 throughout the  process--- adding the first edge of $\cM'$ requires us to delete its two incident edges in $\cM$, which already leads to a deficit of 2 units.

\subsection{Weighted matchings} \label{tightnesswe}
In the weighted case, quantifying this deficit throughout the process is more subtle, but the worst-case scenario remains essentially the same:
$|\cM|= |\cM'|$, all edges have weight $W$ and $H = \cM \oplus \cM'$ is a simple alternating cycle that consists of all edges in $\cM$ and $\cM'$. 
Throughout any transformation process and until handling the last edge of the cycle, 
it must be that $w(\cM^*) \le w(\cM) - W$ if $\Delta < 2|\cM|$. 
In general, the deficit to the weight of the matching is inverse linear in $\Delta$, hence taking $\Delta$ to be $\Theta(1/\eps)$ ensures that the weight of the matching throughout the process never goes below $(1-\eps)w(\cM)$.  
Interestingly, here a similar situation occurs also when $w(\cM') > w(\cM)$. 
Specifically, consider the same example as above but add $\eta$ to each edge weight of $\cM'$, i.e.,  
 assume that the edge weights of $\cM$ and $\cM'$ are now $W$ and $W'  = W + \eta$, respectively.  
Then 
the deficit is no longer $W$ as before (for $1 < \Delta < 2|\cM|$), 
but rather $W - \lfloor \Delta/2 \rfloor  \cdot \eta$ (for $1 < \Delta < 2|\cM|$).
Indeed, adding the first edge of $\cM'$ requires the deletion of its two incident edges in $\cM$, at which stage the deficit is $W-\eta$; from that moment onwards, a single edge of $\cM$ is deleted so that another edge of $\cM'$ can be added, which reduces $\eta$ from the deficit  each time.
Therefore, if $\Delta \cdot \eta = W$, the deficit is always at least $W/2$, 
while $w(\cM') > w(\cM) + W/2$.  
This scenario shows that the bound of Theorem \ref{th:main} is asymptotically tight.
\vspace{5pt}
\\
\noindent 
{\bf Remark.}
If $\Delta = 1$, we'll be in deficit of $2W$ (rather than $W$) throughout the process, similarly to the unweighted case.
In the degenerate case that $\cM$ and $\cM'$ consist of a single edge each and $\Delta = 1$, the weight after the first edge deletion reduces to 0.
\section{Prior work on Dynamic Matching Algorithms} \label{111}
In this appendix we provide a concise literature survey on dynamic approximate matching algorithms.  
(For a more detailed account, see \cite{OR10,BGS11,PS16, Sol16,BHN17,GLSSS19}, and the references therein.)

There is a large body of work on   algorithms for maintaining large matchings with low \emph{amortized} update time,
but none of the papers in this context provides a low \emph{worst-case} recourse bound.
For example, in FOCS'14 Bosek et al.\ \cite{BLSZ14} showed that one can maintain a $(1+\eps)$-MCM in the incremental vertex update   setting
with a total of $O(m / \eps)$ time and $O(n / \eps)$ matching replacements, where $m$ and $n$ denote the number of edges and vertices of the final graph, respectively;
moreover, in SODA'19 \cite{GLSSS19} this result was extended to the incremental edge update setting, achieving a constant (depending exponentially on $\eps$) amortized update time.
While the algorithms of \cite{BLSZ14,GLSSS19}
yield constant (depending on $\eps$) amortized recourse  bounds, no nontrivial (i.e., $o(n)$) worst-case recourse bound is known for the incremental vertex or edge update settings.
As another example, in STOC'16 Bhattacharya et al.\ \cite{BHN16} presented an algorithm for maintaining a $(2+\eps)$-MCM in the fully dynamic edge update setting with an amortized update time $\poly(\log n, 1/\eps)$. While the amortized recourse bound of the algorithm of \cite{BHN16} is dominated by its amortized update time,
$\poly(\log n, 1/\eps)$,  no algorithm for maintaining $(2+\eps)$-MCM with similar  amortized update time and nontrivial worst-case recourse bound is known.

We next focus on algorithms with low \emph{worst-case} update time.

In STOC'13 \cite{NS13} Neiman and Solomon presented an algorithm for maintaining a maximal matching with a worst-case update time $O(\sqrt{m})$,
where $m$ is the dynamic number of edges in the graph. A maximal matching provides a 2-MCM.
\cite{NS13} also provides a 3/2-MCM algorithm with the same update time.
The algorithms of \cite{NS13} provide a constant worst-case recourse bound.
Remarkably, all other dynamic matching algorithms for general and bipartite graphs (described next) do not provide any nontrivial worst-case recourse bound.

In FOCS'13, Gupta and Peng \cite{GP13} presented a scheme for maintaining approximate maximum matchings in fully dynamic graphs,
yielding algorithms for maintaining $(1+\eps)$-MCM with worst-case update times of $O(\sqrt{m} / \eps^2)$ and $O(\Delta / \eps^2)$ in general graphs
and in graphs with degree bounded by $\Delta$, respectively.
The scheme of \cite{GP13} was refined in SODA'16 by Peleg and Solomon \cite{PS16} to provide a worst-case update time of $O(\alpha / \eps^2)$ for graphs with {arboricity} bounded by $\alpha$. (A graph $G=(V,E)$ has \emph{arboricity} $\alpha$ if $\alpha=\max_{U\subseteq V}\left\lceil\frac{|E(U)|}{|U|-1}\right\rceil$, where $E(U)=\left\{(u,v)\in E\mid u,v\in U\right\}$.) Since the arboricity of any $m$-edge graph is $O(\sqrt{m})$ and at most twice the maximum degree $\Delta$, the result of \cite{PS16} generalizes \cite{GP13}.
In ICALP'15 Bernstein and Stein \cite{BS15} gave an algorithm for maintaining a $(3/2+\eps)$-MCM for bipartite graphs with a worst-case update time of $O(m^{1/4}/\eps^{2.5})$; to achieve this result, they employ the algorithm of \cite{GP13} for graphs with degree bounded by $\Delta$, for $\Delta = O(m^{1/4} \sqrt{\eps})$.


The algorithms of \cite{NS13,GP13,BS15,PS16} are deterministic. 
Charikar and Solomon \cite{CS18} and Arar et al.\ \cite{ACCSW18}
independently presented randomized algorithms for maintaining a $(2+\eps)$-MCM with a worst-case update time of $\poly(\log n, 1/\eps)$, assuming an oblivious adversary. Recently Wajc \cite{DBLP:journals/corr/abs-1911-05545} strengthened the results of \cite{CS18,ACCSW18} by achieving similar bounds against an adaptive adversary.
The algorithms of \cite{ACCSW18,DBLP:journals/corr/abs-1911-05545} employs the algorithm of \cite{GP13} for graphs with degree bounded by $\Delta$, for $\Delta = O(\log n / \eps^2)$.

The drawback of the scheme of \cite{GP13} is that the worst-case recourse bounds of algorithms that follow it may be linear in the matching size.
The algorithms of \cite{GP13,BS15,PS16,ACCSW18,DBLP:journals/corr/abs-1911-05545} either follow the scheme of \cite{GP13} or employ one of the aforementioned algorithms provided in \cite{GP13} as a black-box,
hence they all suffer from the same drawback.
Although the algorithm of  \cite{CS18} does not use \cite{GP13} in any way, it also suffers from this drawback, as discussed in Appendix \ref{121}.

\section{Proof of Theorem~\ref{th:main}}
\label{sec:wema}
The setup is as follows.
Let $\cM$ and $\cM'$ be two matchings for the same weighted graph $G = (V,E,w)$. We denote $$w(\cM) = \sum_{e\in \cM} w(e),$$ and assume in what follows that $\cM'$ is an improvement over $\cM$, i.e., that $w(\cM') > w(\cM)$.\footnote{The assumption that $w(\cM') > w(\cM)$ is merely for simplicity of presentation, since in case $w(\cM') \le w(\cM)$, one can gradually transform $w(\cM')$ into $w(\cM)$, and finally reverse the transformation; the weight throughout the process would be at least $\max\{w(\cM')-W',(1-\eps)w(\cM')\}$ with $W' = \max_{e \in \cM'} w(e)$.}
Our goal is to gradually transform $\cM$ into (a possibly superset of) $\cM'$ via a sequence of constant-time operations to be described next, such that the matching obtained at any point throughout this transformation process is a valid matching for $G$ of weight at least $w(\cM)-W$, where $W = \max_{e \in \cM} w(e)$, and also at least $(1-\e)w(\cM)$. It is technically convenient to denote by $\cM^*$ the transformed matching, which is initialized as $\cM$ at the outset, and being gradually transformed into $\cM'$; we refer to $\cM$ and $\cM'$ as the source and target matchings, respectively. 

We achieve this goal in two steps. In the first step (Theorem \ref{th:app}) we show that the weight of the transformed matching never goes below $w(\cM)-W$, and in the second step (Theorem \ref{th:main}) we show that the weight never goes below $\max\{w(\cM)-W,(1-\e)w(\cM)\}$.

Though the proof of Theorem~\ref{th:app} is technically involved, the idea behind it is quite intuitive, and we believe it is instructive to give a high-level overview before getting into the technical details of the proof.
The first observation is that $\cM \cup \cM'$ is a relatively easy-to-analyze graph. Specifically, it is a union of vertex-disjoint {\em $\cM-\cM'$ alternating} simple paths and cycles, except for isolated vertices that we may ignore;
a path in $G$ is called {\em $\cM-\cM'$ alternating} if it consists of an edge in $\cM$ followed by an edge in $\cM'$ and so on, or vice versa.
As $w(\cM')$ is assumed to be greater than $w(\cM)$, we can then efficiently find an ``improving path'' $\Pi$ in $\cM \cup \cM'$, in the sense that the total weight of the edges in $\cM' \cap \Pi$ is greater than the total weight of the edges in $\cM \cap \Pi$. We then show that one can efficiently find a ``minimum vertex'' on $\Pi$ with the following property. We can walk in one direction from that minimum vertex in a cyclic order (even if $\Pi$ is not a cycle) along $\Pi$, deleting the edges of $\cM$ along $\Pi$ and adding the edges of $\cM'$ along $\Pi$ essentially ``one-by-one'' (the formal details appear below).
This will only increase the weight of $\cM^*$, except for a possible \emph{small} loss that we refer to as the \emph{deficit}, which is bounded above by $W = \max_{e \in \cM} w(e)$. We are now ready to state and prove Theorem~\ref{th:app}. 

\begin {theorem}
\label{th:app}
One can gradually transform $\cM$ into (a possibly superset of) $\cM'$ via a sequence of phases, each running in constant time (and thus making at most a constant number of changes to the matching), such that the matching obtained at the end of each phase throughout this transformation process is a valid matching for $G$ of weight at least $w(\cM)-W$, where $W = \max_{e \in \cM} w(e)$.
The runtime of this transformation procedure is $O(|\cM| + |\cM'|)$.
\end {theorem}

\begin {proof}
Recall that the transformed matching, denoted by $\cM^*$, is initialized to the source matching $\cM$.
An edge in $\cM' \setminus \cM^*$ is \emph{good} if the sum of the weights of the edges in $\cM^*$ that are adjacent to it is no greater than its own weight; otherwise it is \emph{bad}.
(These notions generalize the respective notions from Section \ref{proofmain} for unweighted matchings.)
Handling good edges is easy: For any good edge $e$, move it from $\cM' \setminus \cM^*$ to $\cM^*$, and then delete from $\cM^*$ the at most two edges adjacent to $e$;
thus the weight of $\cM^*$ cannot decrease as a result of handling a good edge.
Since every edge in $\cM$ is deleted at most once from $\cM^*$, and every edge in $\cM'$ is added at most once to $\cM^*$, the total time of handling the good edges throughout the algorithm is $O(|\cM|+ |\cM'|) = O(n)$.

Next we describe an algorithm that proves Theorem \ref{th:app}. During the execution of this algorithm, some edges are moved from $\cM' \setminus \cM^*$ to $\cM^*$,
which triggers the removal of edges from $\cM^*$, and as a result some bad edges become good. Similarly to the treatment of Section \ref{proofmain} for the unweighted case, we can use a data structure for maintaining the good edges in $\cM' \setminus \cM^*$,
so that testing if there is a good edge in $\cM^*$ and returning an arbitrary one if exists can  be done in $O(1)$ time.


Just as in the unweighted case, here too we always try to handle good edges as described above, as long as one exists.
The difference between the unweighted case and the weighted one is in how bad edges are handled: In the unweighted case bad edges are handled greedily (in an obvious way), whereas the weighted case calls for a much more intricate treatment, as described next.


We believe it is instructive to refer to the edges of $\cM'$ as \emph{red edges} and to those of $\cM$ as \emph{blue edges}, and our transformation will delete blue edges from $\cM^*$ (recall that $\cM^*$ is initialized to $\cM$) and copy red edges from $\cM'$ to $\cM^*$ so that the invariant in Theorem~\ref{th:app} is always maintained.

We denote the symmetric difference of two sets $A$ and $B$ by $A \oplus B$, i.e., $A \oplus B = (A \cup B) \setminus (A \cap B)$. We use the following well-known observation of Hopcroft and Karp~\cite{HK73}.

\begin {lemma}
\label{le:h}
$H := \cM \oplus \cM'$ is a union of vertex-disjoint alternating blue-red (simple) paths and alternating blue-red (simple) cycles.
\end {lemma}



The \emph{colored weight} of a subgraph $G' \subset G$, denoted by $c(G')$, is defined as the difference between the sum of weights of the red edges in $G'$ and the sum of weights of the blue edges in $G'$.

Since $w(\cM') > w(\cM)$, we have $c(H) > 0$, hence the sum of the colored weights of the alternating blue-red paths and cycles in $H$ is positive.
The following lemma shows a  reduction from a general $H$ to the case where
$H$ is a simple blue-red path or cycle of positive colored weight.
\begin {lemma}
\label {le:red}
If $w(\cM') > w(\cM)$, we may assume that $H$ is an alternating blue-red path or cycle of positive colored weight.
\end {lemma}

\begin {proof}
Let $\Pi_1,\ldots, \Pi_r$ denote the simple alternating blue-red paths or cycles in $H$ ordered so that the paths and cycles of positive colored weight appear first,  and only then the paths or cycles of non-positive colored weight. In other words, the paths and cycles of positive colored weight have smaller indices than those of non-positive colored weight.
As $\sum_{j=1}^r c(\Pi_j) = w(\cM') - w(\cM^*) > 0$ and $c(\Pi_j)$ are ordered positive first, it easily follows that
\begin {equation}
\label{eq: aug}
\sum_{j=1}^i c(\Pi_j) > 0 \quad \text{ for every } 1 \le i \le r .
\end {equation}
Therefore, after treating $\Pi_1, \ldots, \Pi_i$, the weight of $\cM^*$ has increased by $\sum_{j=1}^i c(\Pi_j)$, which is a positive value. When treating the next path $\Pi_{i+1}$ in $H$, we add the non-negative value $\sum_{j=1}^i c(\Pi_j)$ to the weight of the first red edge in $\Pi_{i+1}$, which allows us to view $\Pi_{i+1}$ as 
 having a positive colored weight.
 To complete the proof of this lemma,
  we note that the total transformation is obtained by concatenating all the transformations of $\Pi_1,\ldots, \Pi_r$ in the order in which they appear; and, in particular, if each of these transformations is carried out in time linear in the length of the path/cycle,
	the total runtime of the transformation will be $O(|\cM| + |\cM'|)$.
\QED
\end {proof}

By Lemma~\ref{le:red}, we may assume that $H$ is a path $\Pi$ consisting of $k$ pairs of blue-red edges $b_1 \circ r_1 \circ \cdots \circ b_k \circ r_k$, where $b_i \in \cM^*, r_i \in \cM'$, for $i=1,\ldots, k$, and we allow for $b_1$ or $r_k$ (or both) to be empty; we will make this assumption henceforth.


The algorithm
iteratively changes $\cM^*$ by deleting from $\cM^*$, in iteration $i$, the blue edges $b_{i+1}$ and $b_i$ (if not previously deleted from $\cM^*$), for $i=1,\ldots, k$, thus allowing for an addition of the red edge $r_i \in \cM'$ to $\cM^*$; thus at most 3 changes to $\cM^*$ are made in each iteration. 
 As this basic procedure is used repeatedly below, we include its pseudo-code, and note
that implementing it can be easily done in time $O(|\Pi|)$.  We also note that Procedure $ReplaceBlueRed(\cM^*, \cM', \Pi)$ changes the auxiliary matching $\cM^*$, and by keeping all the intermediate values of $\cM^*$ throughout the process, over all the paths $\Pi$ in $H$, we obtain the entire transformation of $\cM$ into $\cM'$. \\
\begin{mdframed}[backgroundcolor=lightgray!40,topline=false,rightline=false,leftline=false,bottomline=false,innertopmargin=2pt]
Procedure $ReplaceBlueRed(\cM^*, \cM', \Pi)$: \vspace{5pt}\\
{\bf For} $i = 1$ to $k$:
\begin{enumerate}
\item Delete the blue edge $b_{i+1}$ (if exists) and the blue edge $b_{i}$ (if exists), from $\cM^*$.
\item Add the edge $r_{i+1}$ to $\cM^*$.
\end{enumerate}
\end{mdframed}

For a path $\Pi = b_1 \circ r_1 \circ \cdots \circ b_k \circ r_k$, we denote by $\Pi_i =  b_1 \circ r_1 \circ \cdots \circ b_i \circ r_i$ the \emph{alternating blue-red $i$-prefix} of $\Pi$, for any $1\le i \le k$, which is the subpath of $\Pi$ that consists of the first $i$ pairs of blue-red edges.
The colored weight of the alternating blue-red $i$-prefix of $\Pi$ is given by
$$c(\Pi_i) = c\left(b_1 \circ r_1 \circ \cdots \circ b_i \circ r_i \right) = \sum_{j=1}^i w(r_j)-w(b_j),$$ for $0 \le i \le k$, with the convention that $\Pi_0$ is the empty path and $c(\Pi_0)=0$.
Let
\begin {equation}
\label{eq:imin}
i_{min} = \underset {i \in \{0,1,\ldots, k\}}{\arg\min c(\Pi_i)}.
\end {equation}
Clearly, $i_{min}$ and $c(\Pi_{i_{min}})$ can be computed in $O(|\Pi|)$ time.

\paragraph*{First case: $c(\Pi_{i_{min}}) \ge 0$.\\}
In this case we can add the red edges of $\Pi$ and delete the blue edges of $\Pi$ one after another by making a call to Procedure $ReplaceBlueRed(\cM^*, \cM', \Pi)$.
Concretely, after $i$ iterations of the {\bf for} loop in Procedure $ReplaceBlueRed(\cM^*, \cM', \Pi)$, for $1 \le i \le k$, the value of $\cM^*$ is changed by $$c(\Pi_i) - w(b_{i+1}) \ge c(\Pi_{i_{min}}) - w(b_{i+1}) \ge -W,$$
hence the value of $\cM^*$ never decreases by more than $W$. Moreover,
by our assumptions that $H = \Pi$ and $w(\cM') > w(\cM)$, at the end of the execution of the procedure the value of $\cM^*$ is changed by $c(\Pi) = w(\cM') - w(\cM') > 0$.
This shows that the invariant of Theorem~\ref{th:app} is always maintained.

\paragraph*{Second case: $c(\Pi_{i_{min}})  < 0$.\\}
Since $c(\Pi_k) = c(\Pi) > 0$, it follows that $0 < i_{min} < k$. Let $\Pi_{pref} := \Pi_{i_{min}}$ denote the
alternating blue-red $i_{min}$-prefix of $\Pi$, namely $b_1 \circ r_1 \circ \cdots \circ b_{i_{min}} \circ r_{i_{min}}$.
Similarly, we define the \emph{alternating blue-red $i$-suffix} of $\Pi$, for any $1\le i \le k-1$,
as the subpath of $\Pi$ that consists of the last $k - i$ pairs of blue-red edges,
and let $\Pi_{suf}$ denote the alternating blue-red $i_{min}$-suffix of $\Pi$,
namely, $b_{i_{min} + 1} \circ r_{i_{min} + 1} \circ \cdots \circ b_k \circ r_k$. Since
$0 < c(\Pi) = c(\Pi_{pref}) + c(\Pi_{suf})$ and $c(\Pi_{pref}) = c(\Pi_{i_{min}})  < 0$,
we have $c(\Pi_{pref}) < 0 < c(\Pi_{suf}).$
Clearly, $\Pi_{pref}$ and  $\Pi_{suf}$  can be computed in $O(|\Pi|)$ time.

We add to $\cM^*$ the red edges in $\Pi_{suf}$ and delete from $\cM^*$ the blue edges in $\Pi_{suf}$ by making a call to Procedure $ReplaceBlueRed(\cM^*, \cM', \Pi_{suf})$.
After this execution,
the weight of $\cM^*$ is increased by $c(\Pi_{suf})>0$.
Moreover, by the definition of $i_{min}$, after $i$ iterations of the {\bf for} loop of the procedure,
for  $1 \le i \le k-i_{min}$, the value of $\cM^*$ may decrease by at most $w(b_{i+1})\le W$, as compared to its value prior to the call to Procedure $ReplaceBlueRed(\cM^*, \cM', \Pi_{suf})$.

After handling the edges in $\Pi_{suf}$, the value of $\cM^*$ is changed by a value of $c(\Pi_{suf})>0$. Next, we add to $\cM^*$ the remaining red edges in $\Pi$ and delete from $\cM^*$ the remaining blue edges of $\Pi$, by making the call to Procedure $ReplaceBlueRed(\cM^*, \cM', \Pi_{pref})$.
After performing $i$ iterations of the {\bf for} loop of Procedure $ReplaceBlueRed(\cM^*, \cM', \Pi_{pref})$, for $1 \le i \le i_{min}$, the value of $\cM^*$ is changed by at least
\begin {equation}
\begin {array}{ll}
c(\Pi_{i}) - w(b_{i+1}) ~\ge~ c(\Pi_{i}) ~-~ W  \ge c(\Pi_{i_{min}}) - W ~=~ c(\Pi_{pref}) - W ~>~ - c(\Pi_{suf}) - W,\end {array}
\end {equation}
as compared to its value \emph{after} the call to Procedure $ReplaceBlueRed(\cM^*, \cM', \Pi_{suf})$, where the last inequality follows as $c(\Pi_{pref}) + c(\Pi_{suf}) > 0$.
Since we made the call to Procedure $ReplaceBlueRed(\cM^*, \cM', \Pi_{pref})$ after the value of $\cM^*$ has already increased by $c(\Pi_{suf})>0$, it follows that during the execution of Procedure $ReplaceBlueRed(\cM^*, \cM', \Pi_{pref})$, the value of $w(\cM^*)$ may not decrease by more than $W$, as compared to the value of $w(\cM^*)$ \emph{prior} to handling the edges along $\Pi$.

To conclude, in both cases, the value of $w(\cM^*)$ never decreases by more than $W$ during the process of handling the edges along $\Pi$,
and at the end of this process the value grows by $c(\Pi)>0$.
Moreover, the runtime of this procedure is $O(|\Pi|)$.
This completes the proof of Theorem \ref{th:app}.
\QED
\end {proof}

Theorem~\ref{th:app} defines a gradual transformation of $\cM$ into $\cM'$, composed of a sequence of operations of insertions of edges from $\cM'$ into $\cM^*$ and deletions of edges of $\cM$ from $\cM^*$, where each operation makes at most 3 changes to $\cM^*$. By cleverly partitioning this transformation into phases 
consisting of $O(\frac 1 \e)$ operations each,
we can prove the strengthening of Theorem~\ref{th:app} given in Theorem \ref{th:main}, asserting that the weight of the transformed matching at the end of each phase is smaller than the original weight (of $\cM$) not only by at most $W$,
but also by at most a factor of $1-\e$.
\begin {proof} (Proof of Theorem \ref{th:main}.)
If $\max\{w(\cM)-W,(1-\e)w(\cM)\} = w(\cM)-W$, then the theorem follows immediately from Theorem~\ref{th:app}.
It is henceforth assumed that $(1-\e)w(\cM) \ge w(\cM)-W$.

We consider the transformation used in the proof of Theorem~\ref{th:app}, and show how to partition it into phases consisting of $O(\frac 1 \e)$ operations each,
such that the condition in the theorem holds. We keep the same notation as in the proof of Theorem~\ref{th:app}, and use the same reduction of Lemma~\ref{le:red} to the case where $H = \cM^* \oplus \cM'$ is an alternating blue-red path or cycle, denoted by $\Pi$.

By definition of $i_{min}$ (via~(\ref{eq:imin}), within the proof Theorem~\ref{th:app}), we have
\begin {equation}
\label{eq:mat}
c(\Pi_{i+1}) = c(\Pi_i) + w(r_{i+1}) - w(b_{i+1})  \ge c(\Pi_i), \text{  for } i \ge i_{min}.
\end {equation}
Therefore, $w(r_i) \ge w(b_i) \text{  for } i > i_{min}.$

It is possible that many consecutive blue edges $b_{i}$ along $\Pi$ are ``heavy'', i.e., of weight $>\e \cM$,
which is why we partition the transformation into phases, where the specific partition  depends on whether the specific sub-path that we deal with
 (either $\Pi, \Pi_{pref}$, or $\Pi_{suf})$) contains many consecutive heavy blue edges. The following analysis applies to Procedure $ReplaceBlueRed(\cM^*. \cM', \Pi)$,
 but
  the exact same analysis applies verbatim to Procedure $ReplaceBlueRed(\cM^*. \cM', \Pi_{suf})$ or Procedure $ReplaceBlueRed(\cM^*. \cM', \Pi_{pref})$, with the only difference being that the value of $i_{min}$ is set as 0 for the calls of Procedure $ReplaceBlueRed(\cM^*. \cM', \Pi)$ and Procedure $ReplaceBlueRed(\cM^*. \cM', \Pi_{pref})$, and is defined via~(\ref{eq:imin}) for the call to Procedure $ReplaceBlueRed(\cM^*. \cM', \Pi_{suf})$.

We distinguish between two cases.
\paragraph*{First case: $\forall i_{min} < i \le i_{min} + \frac 1 \e:  \ w(b_i) \ge \e w(\cM)$.\\}
In this case the number of edges in $\Pi$ is $\le \frac 2 \e + 1$, as there can be at most $\frac 1 \e$ edges of $\cM$ of weight greater or equal to $\e \cM$. We can thus perform Procedure $ReplaceBlueRed(\cM^*, \cM', \Pi)$ in a single phase of $O(\frac 1 \e)$ operations, following which the value of $\cM^*$ grows by $c(\Pi)>0$.
Clearly, the runtime is $O(|\Pi|)$.

\paragraph*{Second case: $\exists i_{min} < i_0 \le i_{min} + \frac 1 \e: \ w(b_{i_0}) < \e w(\cM)$.\\}
In this case we perform the first $i_0 - i_{min}$ iterations of the {\bf for} loop in Procedure $ReplaceBlueRed(\cM^*, \cM', \Pi)$ in one phase of $O(\frac 1 \e)$ operations. This deletes the blue edges from $\cM^*$ and adds the red edges of $\cM'$ in the sub-path $b_{i_{min}} \circ r_{i_{min}} \circ \cdots \circ b_{i_0-1} \circ r_{i_0 - 1}$ of $\Pi$, within one phase. At this point, we have deleted from $\cM^*$ the edges $b_{i_{min}}, \ldots, b_{i_0}$, and added to $\cM^*$ the edges $r_{i_{min}}, \ldots, r_{i_0-1}$. (Indeed, in step (i) of iteration $i_0-1$ of Procedure $ReplaceBlueRed(\cM^*, \cM', \Pi)$, the edge $b_{i_0}$ is deleted.)

The value of $\cM^*$ has thus changed by $c(\Pi_{i_0})-c(\Pi_{i_{min}}) - w(b_{i_0}) \ge - w(b_{i_0}) $, where the inequality follows by~(\ref{eq:mat}). Therefore, the value of $\cM^*$ decreased by at most $w(b_{i_0})$, which is smaller than $\e \cM$ by definition of $i_0$, and thus $\cM^*$ has weight at least $(1-\e)w(\cM)\}$
at the end of this phase.

We repeat this process recursively, treating the remaining edges in $\Pi$ with the same bifurcation into two cases,
thus maintaining the invariant that at the end of each phase throughout this transformation process, $\cM^*$ is a valid matching with $w(\cM^*) \ge (1-\e)w(\cM) \ge \max\{w(\cM)-W, (1-\e)w(\cM)\}$.
Here too the runtime is $O(|\Pi|)$.
\ignore{
In the first case, we proceed to perform the $ReplaceBlueRed(\cM^*, \cM', \Pi_{rem})$ in one batch for $O(\frac 1 \e)$ steps, and increase the value of $\cM^*$ by $c(\Pi)>0$, without further decreasing $\cM^*$. In the second case, we find $i_0 < i_1 \le i_0 + \frac 1 \e$ with $w(b(i_1)) < \e w(\cM)$, and continue to perform the next $i_1 - i_0 + 1$ iterations of (i), (ii) and (iii) in $ReplaceBlueRed(\cM^*, \cM', \Pi)$ in one batch of $O(\frac 1 \e)$ steps. This will delete the blue edges from $\cM^*$ and add the red edges of $\cM^*$ in $b_{i_0} \circ r_{i_0} \circ \cdots \circ b_{i_1-1} \circ r_{i_1 - 1}$  according to their appearance in $\Pi$, within one batch. At this point, since the beginning of the execution of $ReplaceBlueRed(\cM^*, \cM', \Pi)$, we have deleted from $\cM^*$ the edges $b_{i_{min}}, \ldots, b_{i_1}$ and added to $\cM^*$ the edges $r_{i_{min}}, \ldots, r_{i_1-1}$. The value of $\cM^*$ has increased (since the beginning of treating $\Pi$) by $c(\Pi_{i_1})-c(\Pi_{i_{min}}) - w(b_{i_1}) \ge - w(b_{i_1})$, where the inequality again follows by~(\ref{eq:mat}). Therefore, the value of $\cM^*$ decreased by at most $w(b_{i_1}) < \min\{\e \cM, W\}$, and we can repeat this partition into batches of size $O(\frac 1 \e)$ of the transformation in $ReplaceBlueRed(\cM^*, \cM', \Pi)$.}

\paragraph* {Runtime analysis.}
We partition the edges of $H = \cM \cup \cM'$ into alternating blue-red simple paths and cycles, and order them so that the paths and cycles of positive colored weight appear first, which can be easily done in time $O(|\cM| +|\cM'|)$. The treatment of each alternating path or cycle $\Pi$ in $H$ is done by making a call to Procedure $ReplaceBlueRed(\cM^*, \cM', \Pi)$, where we partition the resulting transformation into batches of size $O(\frac 1 \eps)$ each;  summing over all paths in $H$, this can also be easily done in time $O(|\cM| +|\cM'|)$. 
The total time required for handling the good edges of $\cM'$ throughout the algorithm is $O(|\cM| +|\cM'|)$, hence the running time of the transformation   is $O(|\cM| +|\cM'|)$.

\vspace{6pt}
The proof of Theorem~\ref{th:main} follows.
\QED
\end {proof}

\begin {remark}
Recall that in the unweighted case, when $|\cM| < |\cM'|$, it was possible to gradually transform $\cM$ to $\cM'$ without ever being in deficit compared go the initial value of $\cM$, i.e., $|\cM^*| \ge |\cM|$ throughout the transformation process.
However, if $|\cM'| \le |\cM|$, we showed in App.\ \ref{tightnessun}
that this is no longer the case and there are cases in which we are always at a deficit of at least one unit until the very end of the transformation process. 
In the weighted case, there are examples where $w(\cM') > w(\cM)$, yet in every gradual transformation of $\cM$ to $\cM'$, we have $w(\cM^*) < w(\cM)$ throughout the transformation. 
The deficit 
throughout the process is more subtle to quantify, and this was
discussed in detail in App.\ \ref{tightnesswe}.
\end {remark}

\section{Proof of Lemma \ref{lazylemma2}} \label{app:lazy}
Write $k = \lfloor \eps'\cdot |\cM_t| \rfloor$, and let $k_{ins}$ and $k_{del}$ denote the number of (edge or vertex) insertions and deletions that occur during the $k$ updates $t+1,\ldots,t+k$,
respectively, where $k = k_{ins} + k_{del}$. We have $|\cM^{OPT}_t| \le \beta \cdot |\cM_t|$, where $\cM^{OPT}_i$ is a maximum matching for $G_i$, for all $i$.
Since each  insertion may increase the size of the MCM by at most 1, we have
$$|\cM^{OPT}_i| ~\le~ |\cM^{OPT}_t| + k_{ins} ~\le~ \beta \cdot |\cM_t| + k ~\le~ \beta \cdot |\cM_t| (1 + \eps').$$
Also, each   deletion may remove at most one edge from $\cM_t$, hence
$$|\cM^{(i)}_t| ~\ge~ |\cM_t| - k_{del} ~\ge~ |\cM_t| - k ~\ge~ |\cM_t| - \eps' \cdot |\cM_t|  ~=~ |\cM_t| (1- \eps').$$
It follows that $$\frac{|\cM^{OPT}_i|}{|\cM^{(i)}_t|} ~\le~
\frac{\beta \cdot |\cM_t| (1 + \eps')}{|\cM_t| (1- \eps')} ~\le~ \beta(1 + 2\eps'), $$
where the last inequality holds for all $\eps' \le 1/2$. The lemma follows.

\section{Further details on the scheme of \cite{GP13}} \label{amort}
The key insight behind the scheme of \cite{GP13} and of its generalization \cite{PS16} is not to compute the approximate matching on the entire graph, but rather on a \emph{matching sparsifier}, which is a sparse subgraph $\tilde G$ of the entire graph $G$ that preserves the maximum matching size to within a factor of $1+\eps$.
The matching sparsifier of \cite{GP13, PS16} is derived from a constant approximate minimum vertex cover that is maintained dynamically \emph{by other means}.
We will not describe here the manners in which a constant approximate minimum vertex cover is maintained and the sparsifier is computed on top of it; the interested reader can refer to \cite{GP13,PS16} for details. The bottom-line of \cite{GP13,PS16} is that for graphs with arboricity bounded by $\alpha$, for any $1 \le \alpha = O(\sqrt{m})$,
the matching sparsifier $\tilde G$ of \cite{PS16} has $O(|\cM| \cdot \alpha/\eps)$ edges, and it can be computed in time linear in its size.
(The scheme of \cite{PS16} generalizes that of \cite{GP13}, hence we might as well restrict attention to \cite{PS16}.)
An $(1+O(\eps))$-MCM can be computed for the sparsifier $\tilde G$ in time $O(|\tilde G| / \eps) = O(|\cM| \cdot \alpha/\eps^2)$ \cite{HK73,MV80,Vaz12},
and assuming the constant hiding in the $O$-notation is sufficiently small, it provides a $(1+\eps/4)$-MCM for the entire graph.
Since the cost $O(|\cM| \cdot \alpha/\eps^2)$ of this static computation is amortized over  $\Omega(\eps \cdot |\cM|)$ update steps,
the resulting amortized update time is $O(\alpha \cdot \eps^{-3})$.
As shown in \cite{PS16}, one can shave a $1/\eps$-factor from this update time bound, reducing it to $O(\alpha \cdot \eps^{-2})$,
but the details of this improvement lie outside the scope of this short overview of the scheme of \cite{GP13,PS16}.

\section{Discussion and Open Problems} \label{discuss}
This paper introduces a natural generalization of the MRP, and provides near-optimal transformations to the problems of
maximum cardinality matching and maximum weight matching.

One application of this meta-problem is to dynamic graph algorithms.
In particular, by building on our transformation for maximum cardinality matching we have shown that any algorithm for maintaining a $\beta$-MCM can be transformed into
an algorithm for maintaining a $\beta(1+\eps)$-MCM with essentially the same update time as that of the original algorithm and with a worst-case recourse bound of $O(1/\eps)$, for any $\beta \ge 1$ and $\eps>0$. This recourse bound is optimal for the regime $\beta = 1+\eps$.
We also extended this result for weighted matchings, but there is a linear dependence on the aspect-ratio of the graph
in the update time and recourse bounds. 
It would be interesting to improve this dependency to be polylogaritmic in the aspect-ratio.

It would be interesting to study additional basic graph problems under this generalized framework. Although our positive results may lead to the impression that
there exists an efficient gradual transformation process to any optimization graph problem, we conclude with a sketch of two trivial hardness results.

For the \emph{maximum independent set problem} any gradual transformation process cannot provide any nontrivial approximation guarantee,
regardless of the approximation guarantees of the source and target independent sets.  
To see this, denote the source approximate maximum independent set (the one we start from) by $\cS$ and the target approximate maximum independent set (the one we gradually transform into) by $\cS'$,
and suppose there is a complete bipartite graph between $\cS$ and $\cS'$.
Since we cannot add even a single vertex of $\cS'$ to the output independent set $\cS^*$ (which is initialized as $\cS$) before removing from it all vertices of $\cS$
and assuming each step of the transformation process makes only $\Delta$ changes to $\cS^*$,
the approximation guarantee of the output independent set
must reach $\Omega(|\cS'|/ \Delta)$ at some moment throughout the transformation process. In other words, the approximation guarantee may be arbitrarily large.

As another example, an analogous argument shows that for the \emph{minimum vertex cover problem}, any gradual transformation process cannot provide an approximation guarantee better than
$\frac{|\cC| + |\cC'|}{|\cC'|} > 2$, where $\cC$ and $\cC'$ are the source and target vertex covers, respectively.
On the other hand, one can easily see that the approximation guarantee throughout the entire transformation process does not exceed $\frac{|\cC| + |\cC'|}{|\cC^{OPT}|}$,
where $\cC^{OPT}$ is a minimum vertex cover for the graph,
by gradually adding all vertices of the target vertex cover $\cC'$ to the output vertex cover $\cC^*$ (which is initialized as $\cC$), and later gradually removing the vertices of $\cC$ from the output vertex cover $\cC^*$.

These examples demonstrate a basic limitation of the our generalized framework, and suggest that further research of this framework is required. One interesting direction for further research is studying the maximum independent set and minimum vertex cover problems for bounded degree graphs; note that the trivial hardness results mentioned above do not apply directly to bounded degree graphs.
More generally, studying additional combinatorial optimization problems under this framework may contribute to a deeper understanding of 
its inherent limitations and strengths, and in particular,
to finding additional applications of this framework, possibly outside the area of dynamic matching algorithms.

\end{document}